\documentclass[english,american,12pt, draftclsnofoot, onecolumn]{IEEEtran}
\usepackage[T1]{fontenc}
\usepackage[latin9]{inputenc}
\usepackage{array}
\usepackage{float}
\usepackage{mathrsfs}
\usepackage{amsmath}
\usepackage{amsthm}
\usepackage{amssymb}
\usepackage{graphicx}

\makeatletter

\providecommand{\tabularnewline}{\\}
\floatstyle{ruled}
\newfloat{algorithm}{tbp}{loa}
\providecommand{\algorithmname}{Algorithm}
\floatname{algorithm}{\protect\algorithmname}

\theoremstyle{plain}
\newtheorem{lem}{\protect\lemmaname}
\theoremstyle{plain}
\newtheorem{thm}{\protect\theoremname}
\theoremstyle{plain}
\newtheorem{prop}{\protect\propositionname}

\PassOptionsToPackage{dvipsnames,table}{xcolor}

\usepackage{amsmath,amssymb}

\usepackage{algorithm}
\usepackage{algorithmic}

\usepackage[table]{xcolor}

\usepackage{graphicx,psfrag,cite,subfigure}

\author{
Yuanshuai~Zheng, and
Junting~Chen

}


\usepackage[acronym]{glossaries}
\newcommand{\newac}{\newacronym}
\newcommand{\ac}{\gls}

\newcommand{\acpl}{\glspl}

\makeglossaries
\newac{speb}{SPEB}{square position error bound}
\newac[plural=EFIMs,firstplural=Fisher information matrices (EFIMs)]{efim}{EFIM}{Fisher information matrix}
\newac{ne}{NE}{Nash equilibrium}
\newac{mse}{MSE}{mean squared error}
\newac{toa}{TOA}{time-of-arrival}
\newac{snr}{SNR}{signal-to-noise ratio}
\newac{lan}{LAN}{local area network}
\newac{psd}{PSD}{positive semidefinite}
\newac{pd}{PD}{positive definite}
\newac{wrt}{w.r.t.}{with respect to}
\newac{lhs}{L.H.S.}{left hand side}
\newac{wp1}{w.p.1}{with probability 1}
\newac{kkt}{KKT}{Karush-Kuhn-Tucker}
\newac{wlog}{w.l.o.g.}{without loss of generality}
\newac{mle}{MLE}{maximum likelihood estimation}
\newac{gps}{GPS}{global positioning system}
\newac{rssi}{RSSI}{received signal strength indicator}
\newac{mimo}{MIMO}{multiple-input multiple-output}
\newac{csi}{CSI}{channel state information}
\newac{fdd}{FDD}{frequency division duplexing}
\newac{ms}{MS}{mobile station}
\newac{bs}{BS}{base station}
\newac{d2d}{D2D}{device-to-device}
\newac{slnr}{SLNR}{signal-to-interference-leakage-and-noise-ratio}
\newac{ula}{ULA}{uniform linear antenna array}
\newac{pas}{PAS}{power angular spectrum}
\newac{mmse}{MMSE}{minimum mean square error}
\newac{zf}{ZF}{zero-forcing}
\newac{rzf}{RZF}{regularized zero-forcing}
\newac{as}{AS}{angular spread}
\newac{aod}{AOD}{angle of departure}
\newac{iid}{i.i.d.}{independent and identically distributed} 
\newac{sinr}{SINR}{signal-to-interference-and-noise ratio}
\newac{tdd}{TDD}{time-division duplex}
\newac{rvq}{RVQ}{random vector quantization}
\newac{rhs}{R.H.S.}{right hand side}
\newac{mrc}{MRC}{maximum ratio combining}
\newac{cdf}{CDF}{cumulative distribution function}
\newac{a.s.}{a.s.}{almost surely}
\newac{los}{LOS}{line-of-sight}
\newac{jsdm}{JSDM}{joint spatial division and multiplexing}
\newac{map}{MAP}{maximum a posteriori}
\newac{klt}{KLT}{Karhunen-Lo\`eve Transform}
\newac{lbe}{LBE}{link bargaining equilibrium}
\newac{se}{SE}{Stackelberg equilibrium}
\newac{uav}{UAV}{unmanned aerial vehicle}
\newac{nlos}{NLOS}{non-line-of-sight}
\newac{pdf}{PDF}{probability density function}
\newac{em}{EM}{expectation-maximization}
\newac{knn}{KNN}{$k$-nearest neighbor}
\newac{svd}{SVD}{singular value decomposition}
\newac{nmf}{NMF}{non-negative matrix factorization}
\newac{umf}{UMF}{unimodality-constrained matrix factorization}
\newac{rmse}{RMSE}{rooted mean squared error}
\newac{olos}{OLOS}{obstructed line-of-sight}
\newac{mmw}{mmW}{millimeter wave}
\newac{ber}{BER}{bit error rate}
\newac{rss}{RSS}{received signal strength}
\newac{lp}{LP}{linear program}
\newac{ufw}{U-FW}{unimodal Frank-Wolfe}
\newac{utf}{UTF}{unimodality-constrained tensor factorization}
\newac{fw}{FW}{Frank-Wolfe}
\newac{iot}{IoT}{Internet-of-Things}
\newac{mae}{MAE}{mean absolute error}
\newac{crb}{CRB}{Cram\'er-Rao bound}
\newac{aoa}{AoA}{angle of arrival}
\newac{wcl}{WCL}{weighted centroid localization}

\usepackage{enumitem}

\makeatother

\usepackage{babel}
\addto\captionsamerican{\renewcommand{\lemmaname}{Lemma}}
\addto\captionsamerican{\renewcommand{\propositionname}{Proposition}}
\addto\captionsamerican{\renewcommand{\theoremname}{Theorem}}
\addto\captionsenglish{\renewcommand{\algorithmname}{Algorithm}}
\addto\captionsenglish{\renewcommand{\lemmaname}{Lemma}}
\addto\captionsenglish{\renewcommand{\propositionname}{Proposition}}
\addto\captionsenglish{\renewcommand{\theoremname}{Theorem}}
\providecommand{\lemmaname}{Lemma}
\providecommand{\propositionname}{Proposition}
\providecommand{\theoremname}{Theorem}

\begin{document}
\title{Geography-aware Optimal UAV 3D Placement for LOS Relaying: A Geometry
Approach}

\maketitle
%
%


\newcommand*{\SINGLECOLUMN}{}

\ifdefined\SINGLECOLUMN
	\setkeys{Gin}{width=0.5\columnwidth}
	\newcommand{\figfontsize}{\footnotesize} 
	\newcommand{\labelfontsize}{0.7}
	\newcommand{\legendfontsize}{0.57}
	\newcommand{\ticklabelfontsize}{0.7}
	\newcommand{\numberfontsize}{0.45}
\else
	\setkeys{Gin}{width=1.0\columnwidth}
	\newcommand{\figfontsize}{\normalsize} 
	\newcommand{\labelfontsize}{0.7}
	\newcommand{\legendfontsize}{0.65}
	\newcommand{\ticklabelfontsize}{0.7}
	\newcommand{\numberfontsize}{0.5}
\fi

\begin{abstract}
Many emerging technologies for the next generation wireless network
prefer \ac{los} propagation conditions to fully release their performance
advantages. This paper studies 3D \ac{uav} placement to establish
\ac{los} links for two ground terminals in deep shadow in a dense
urban environment. The challenge is that the \ac{los} region for
the feasible \ac{uav} positions can be arbitrary due to the complicated
structure of the environment. While most existing works rely on simplified
stochastic \ac{los} models and problem relaxations, this paper focuses
on establishing theoretical guarantees for the optimal \ac{uav} placement
to ensure \ac{los} conditions for two ground users in an actual propagation
environment. It is found that it suffices to search a bounded 2D area
for the globally optimal 3D \ac{uav} position. Thus, this paper develops
an exploration-exploitation algorithm with a linear trajectory length
and achieves above $99\%$ global optimality over several real city
environments being tested in our experiments. To further enhance the
search capability in an ultra-dense environment, a dynamic multi-stage
algorithm is developed and theoretically shown to find an $\epsilon$-optimal
UAV position with a search length $O(1/\epsilon)$. Significant performance
advantages are demonstrated in several numerical experiments for wireless
communication relaying and wireless power transfer.
\end{abstract}

\begin{IEEEkeywords}
UAV, Positioning, LOS relaying, Geography-aware, Wireless power transfer
\end{IEEEkeywords}

\glsresetall

\section{Introduction}

\label{sec:intro}

The \ac{los} propagation conditions are desired in many trending
technologies for the next generation wireless networks. For example,
millimeter-wave and terahertz signals have much less diffractive and
reflective paths compared to sub-6GHz signals due to their small wavelengths,
and hence, \ac{nlos} terminals usually suffer from great path loss
\cite{WanHuaWanGao:J20,WanWanHuJia:J21}. Free-space optical signals
are difficult to penetrate obstacles. Wireless power transfer (WPT)
also prefers \ac{los} conditions for energy efficiency. However,
it is challenging to establish \ac{los} conditions in a dense urban
area, where high buildings and trees easily block the signals.

Low altitude \ac{uav} provides a promising solution to establish
\ac{los} links to ground terminals in deep shadow \cite{ZenWuZha:J19,GerGarAzaLoz:J22}.
Recent works have discussed employing \acpl{uav} in many scenarios,
including communication relaying, data collection, coverage extension,
and WPT \cite{ZhoGuoLiChe:J20,MaZhoQiaChe:J21,MoaSaaChaTri:J20,LiYaoWanXu:J19,LiuXioLuNi:J21}.
To mitigate the potential issue of the limited propulsion energy at
the \acpl{uav}, some recent works have proposed solutions including
dynamic service landing spots \cite{WanSuZhaLi:J22} and the applications
of tethered \acpl{uav} \cite{MusAhmMoh:J20,LimYuLee:J22,ZhaLiuAns:J22}.

Despite many existing works on \ac{uav} placement for wireless communications,
very few solutions guarantee to establish \ac{los} conditions for
specific users in deep shadow. Most existing works tend to oversimplify
the terrain environment. For example, some early works \cite{JiaSwi:J12,ZenZhaLim:J16,LyuZenZha:J18}
studied UAV placement using a pure \ac{los} model, assuming no blockage
from the terrain. A probabilistic \ac{los} model for urban environment
was established in \cite{AlhKanLar:J14,GapMolAndHea:J21}, with extensions
in \cite{SamRapMac:J15,EsrGanGes:J19}, and was adopted in \cite{CheMozSaaYin:J17,IreSebHal:J18,CheHua:J22,SinAgrSinBan:J22}
for \ac{uav} placement and trajectory planning. Since the models
\cite{CheMozSaaYin:J17,IreSebHal:J18,CheHua:J22,SinAgrSinBan:J22}
capture the \ac{los} conditions only in a {\em statistical} sense,
the corresponding solutions developed in \cite{CheMozSaaYin:J17,IreSebHal:J18,CheHua:J22,SinAgrSinBan:J22}
cannot guarantee \ac{los} conditions for specific users.

Some recent attempts exploit radio maps or city maps to assist \ac{uav}
placement, where radio maps describe the channel quality between a
ground terminal and a possible UAV position \cite{ZenXuJinZha:J21,ZhaZha:J21}.
Yet, it is still challenging to search for the best UAV position.
The work \cite{ZenXuJinZha:J21} applied deep reinforcement learning
(DRL) to assist the navigation of \ac{uav}, but the optimality and
complexity are difficult to analyze. In \cite{ZhaZha:J21} and \cite{DonHeWanZha:J22},
the authors used signal-to-interference-plus-noise ratio (SINR) map-based
methods to solve the \ac{uav} 3D path planning problem, but these
methods require offline city maps or radio maps, and hence, they are
difficult to be applied to online search. In \cite{DabSad:J20}, a
geometry-based approach was developed to optimize the UAV position
for free-space optical relaying for two ground users, but the approach
only guarantees the optimality in a 2D plane. In \cite{YiZhuZhuXia:J22},
the buildings were approximately modeled as a set of polyhedrons,
and a number of constraints on the \ac{uav} positions were formulated
using geometry relations; accordingly, a non-convex \ac{uav} placement
problem was formed and relaxation-based algorithms were developed,
although the global optimality was still unknown. In summary, the
main challenge of the \ac{uav} placement problem originates from
the fact that the terrain obstacles may have arbitrary locations and
shapes, and therefore, the placement problems are generally non-convex
with possibly arbitrarily many local optima.

In this paper, we attempt to establish some theoretical guarantees
for the optimal \ac{uav} placement to ensure \ac{los} conditions
for two ground users in an almost arbitrary urban environment. The
goal is to develop an efficient search strategy to explore only a
small 2D local area for the best 3D \ac{uav} placement. Some prior
work \cite{CheMitGes:T21} attempted a special case of the problem,
where one of the users is placed on a high tower such that there is
always an \ac{los} link between the user and the \ac{uav}. However,
when both users are on the ground and are likely shadowed by buildings,
the method in \cite{CheMitGes:T21} fails to apply.

This paper exploits two universal properties for any \ac{los} patterns
from an almost arbitrary terrain structure: {\em upward invariance}
and {\em colinear invariance}. Specifically, if a \ac{uav} sees
a user, such an \ac{los} condition will remain if the \ac{uav} increases
its altitude or moves away from the user without changing the elevation
and azimuth angles, under some additional mild conditions. Exploiting
these properties, two search strategies are developed. The key theoretical
results and numerical findings are summarized as follows.
\begin{itemize}
\item We develop a search trajectory, Algorithm~1, on the middle perpendicular
plane of the two users. It is proven that the search finds the optimal
solution on the middle perpendicular plane, and the search length
is upper bounded by a linear function of the altitude of the initial
point.
\item We show that given a double-\ac{los} initial point, it suffices to
search a bounded 2D local area for the globally 3D optimal \ac{uav}
position. With this analytical insight, we develop Algorithm~2 with
search complexity $O(1/\text{\ensuremath{\epsilon}})$ for the $\epsilon$-optimal
\ac{uav} position in 3D under some mild condition.
\item We conduct numerical experiments using real city map data for several
typical cities. It is found that both Algorithms 1 and 2 achieve over
$99\%$ optimality in a moderate dense environment. In a simulated
ultra dense environment based on a street map of Guangzhou, China,
Algorithm~2 can achieve over $98\%$ of the global optimality under
a reasonable search distance.
\end{itemize}

The remaining part of the paper is organized as follows. Section II
introduces the system model, and formulates a geography-aware \ac{uav}
position optimization problem that can be employed in multiple applications.
Section III presents Algorithm~1 for the optimal solution on the
middle perpendicular plane with theoretical proof of the optimality
and linear complexity. In Section IV, we further propose Algorithm~2
based on the extracted geographic features of \ac{los} patterns,
and demonstrate the performance-complexity trade-off. Section V contains
our simulations accompanied by the relevant discussion and comparison,
and finally, the paper is concluded in Section VI.

\section{System Model}

\label{sec:system-model-dual}

\subsection{Blockage-aware Air-to-ground Channel Model}

\label{subsec:Blockage-aware-Air-to-ground-Channel-model}

Consider to place a \ac{uav} to establish \ac{los} channels to two
users located at $\mathbf{u}_{1}$ and $\mathbf{u}_{2}$ on the ground
in an outdoor urban environment. The \ac{uav} is deployed with
a minimum height $H_{\text{min}}$ which is set to be greater than
the tallest structure in the area such that there will be no potential
collision for the \ac{uav}. The users are possibly surrounded by
urban structures, and thus, the wireless communication link between
the \ac{uav} and the user can be blocked by buildings or trees.

For presentation convenience, define a Cartesian coordinate system
with the origin $O$ set at $(\mathbf{u}_{1}+\mathbf{u}_{2})/2$,
and three orthonormal basis vectors $\text{\ensuremath{\mathbf{e}_{1}}}$,
$\mathbf{e}_{2}$, and $\mathbf{e}_{3}$ where $\mathbf{e}_{2}=(\mathbf{u}_{2}-\mathbf{u}_{1})/\|\mathbf{u}_{2}-\mathbf{u}_{1}\|_{2}$
representing the direction from user $1$ to user $2$, $\mathbf{e}_{3}$
is the direction perpendicular to the ground pointing upward, and
$\mathbf{e}_{1}$ is determined according to the right-hand rule as
illustrated in Fig.~\ref{fig:symmetric_trajectory}.

Denote $\mathcal{D}_{0}^{(i)}$ as the set of permissible \ac{uav}
positions $\mathbf{p}=(p_{1},p_{2},p_{3})$ such that there is an
\ac{los} link between the \ac{uav} and the $i$th user, $i=1,2$,
and $p_{3}\geq H_{\text{min}}$. The \ac{los} regions $\mathcal{D}_{0}^{(i)}$
can be {\em arbitrary} except that we assume $\mathcal{D}_{0}^{(i)}$
have the following properties: For any $\mathbf{p}\in\mathcal{D}_{0}^{(i)}$,
\begin{enumerate}
\item Upward invariant: any position $\mathbf{p}'$ perpendicularly above
$\mathbf{p}$ also belongs to $\mathcal{D}_{0}^{(i)}$, \emph{i.e.},
$\mathbf{p}'\in\mathcal{D}_{0}^{(i)}$;
\item Colinear invariant: any position $\mathbf{p}'$ that satisfies $\mathbf{p}'-\mathbf{u}_{i}=\rho(\mathbf{p}-\mathbf{u}_{i})$
for some $\rho>1$ also belongs to $\mathcal{D}_{0}^{(i)}$, \emph{i.e.},
$\mathbf{p}'\in\mathcal{D}_{0}^{(i)}$.
\end{enumerate}

In addition, define {\em double-LOS} region $\tilde{\mathcal{D}}_{0}=\mathcal{D}_{0}^{(1)}\cap\mathcal{D}_{0}^{(2)}$
as the set of \ac{uav} positions where there are \ac{los} links
to both users. Since the {double-LOS} region is an intersection
of $\mathcal{D}_{0}^{(i)}$, the upward invariant property automatically
holds, \emph{i.e.}, for any double-LOS position $\mathbf{p}\in\tilde{\mathcal{D}}_{0}$,
any position $\mathbf{p}'$ perpendicularly above $\mathbf{p}$ is
also a double-LOS position which satisfies $\mathbf{p}'\in\tilde{\mathcal{D}}_{0}$.
Note that the colinear invariant property does not hold for $\tilde{\mathcal{D}}_{0}$.

To summarize, the upward invariant and colinear invariant properties
imply that if a UAV sees a user at $\mathbf{u}_{i}$, such an \ac{los}
condition will remain if the \ac{uav} increases its altitude or moves
away from the user without changing the elevation and azimuth angles.
The widely adopted probabilistic \ac{los} model in the \ac{uav}
literature \cite{CheMozSaaYin:J17,IreSebHal:J18,CheHua:J22,SinAgrSinBan:J22}
is a special case that satisfies these properties in a statistical
sense.

The upward invariant and colinear invariant properties can be easily
understood from the ray-tracing mechanism based on the geometry relation
with the environment. The implication is that if the urban structures
all have their top no wider than the base, for instance, a combination
of straight pillars and cones, then the upward invariant and colinear
invariant properties can be automatically satisfied. While practical
city topologies may occasionally violate these properties in some
local area, these properties still serve as a good approximation to
the radio environment of interest.

\subsection{Geography-aware \ac{uav} Position Optimization}

\label{subsec:Optimization_problem_and_applications}

The goal of this paper is to place the \ac{uav} as close to both
users as possible under the double-\ac{los} condition $\mathbf{p}\in\tilde{\mathcal{D}}_{0}$.
Specifically, denote $f(d_{i}(\mathbf{p}))$ as the value function
in terms of the distance $d_{i}(\mathbf{p})=\|\mathbf{p}-\mathbf{u}_{i}\|_{2}$
from the \ac{uav} position $\mathbf{p}$ to the user $\mathbf{u}_{i}$,
$i\in\{1,2\}$. The function $f(d)$ is assumed to be continuous and
decreasing in $d$. The objective is to maximize the performance of
the worse link under the double-\ac{los} condition: 
\begin{equation}
\begin{aligned}\mathscr{P}:\quad\mathop{\mbox{maximize}}\limits _{\mathbf{p}} & \quad F(\mathbf{p})\\
\mathop{\mbox{subject to}} & \quad\mathbf{p}\in\tilde{\mathcal{D}}_{0}
\end{aligned}
\label{basic_problem}
\end{equation}
where $F(\mathbf{p})=\min\{f(d_{1}(\mathbf{p})),f(d_{2}(\mathbf{p}))\}$.

Typical applications of the above formulation include \ac{uav}-assisted
relay communications, WPT to ground devices, and video monitoring
of two ground spots. In decode-and-forward relaying, for instance,
one may choose $f(d)=B\log_{2}(1+\gamma d^{-\alpha})$, where $B$
is the bandwidth, $\gamma$ is the effective \ac{snr}, and $\alpha$
is the path-loss exponent in \ac{los} (see Section \ref{sec:Numerical-Results}
for a more specific example). For WPT or visual monitoring, the link
performance function can be chosen as $f(d)=\kappa d^{-\alpha}$,
where $\kappa$ and $\alpha$ are some parameters depending on the
applications.

The main challenge is due to the possibly complicated structure of
the double-LOS region $\tilde{\mathcal{D}}_{0}$. First, $\tilde{\mathcal{D}}_{0}$
may appear to have an irregular pattern as shown in Fig.~\ref{fig:extreme-case},
where a good solution may not be found in a straight-forward way.
Fig.~\ref{fig:extreme-case} shows a topology viewed from the top,
where users are surrounded by high buildings. The grid-shaded area
represents the double-\ac{los} region $\tilde{\mathcal{D}}_{0}$
sliced at the altitude $H_{\text{min}}$. In Fig.~\ref{fig:extreme-case}(a),
the double-\ac{los} region is off the middle perpendicular plane
between two users when they are behind tall buildings. In Fig.~\ref{fig:extreme-case}(b),
the double-LOS region can even be far away from the two users when
they are surrounded by tall buildings.

Second, the structure of $\tilde{\mathcal{D}}_{0}$ may lead to arbitrarily
many local optima due to possibly a huge number of structures or sub-structures
in the urban area of interest. As a result, the approach in \cite{YiZhuZhuXia:J22}
which models the environment using polyhedrons is difficult to implement
due to the complexity of the environment and the possibly large amount
of local optima. This paper, on the contrary, attempts to develop
an exploration-exploitation approach with an aim to establish some
theoretical guarantees for the global optimality of the UAV position.
\begin{figure}
\begin{centering}
\subfigure[]{\includegraphics[width=0.25\columnwidth]{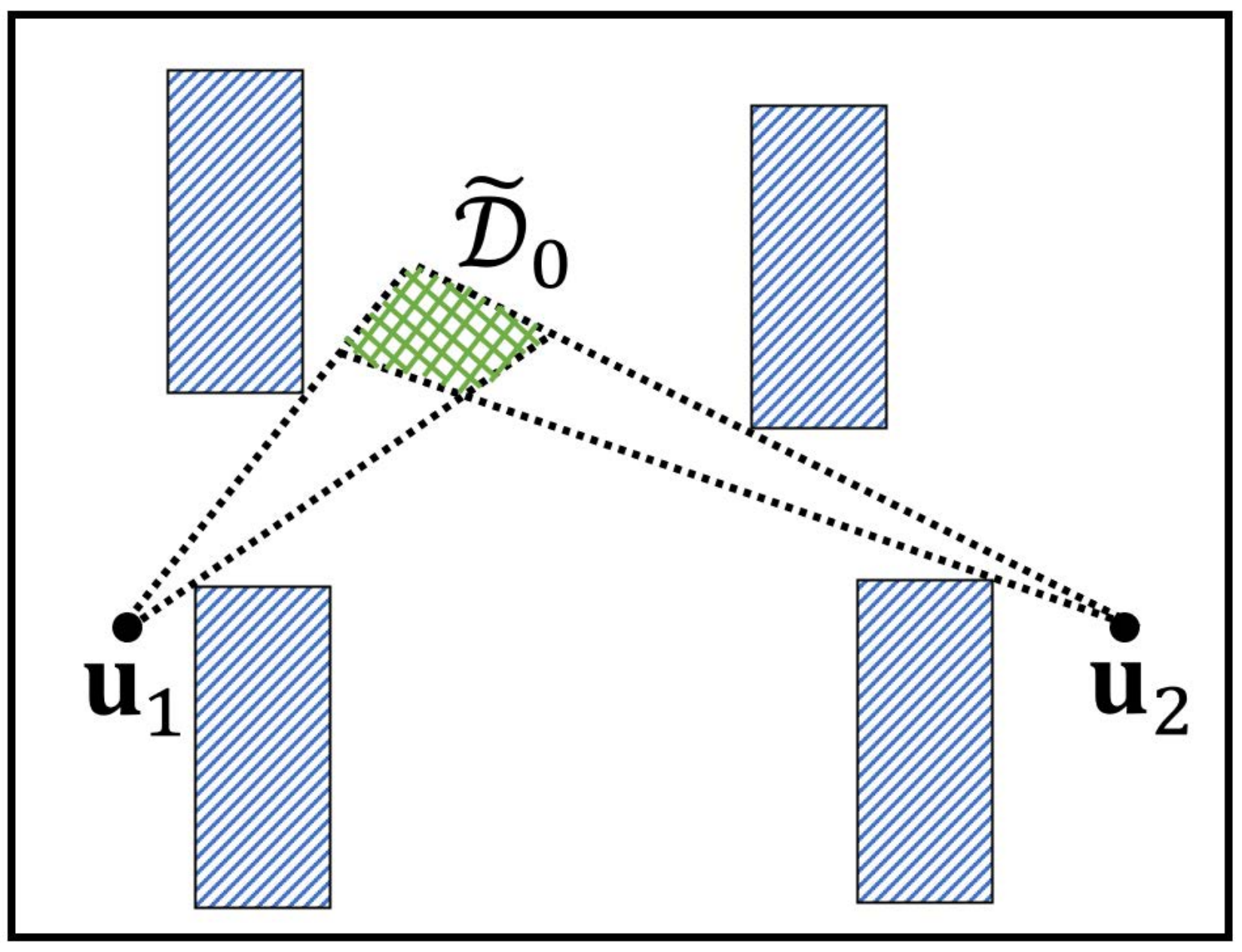}}\subfigure[]{\includegraphics[width=0.25\columnwidth]{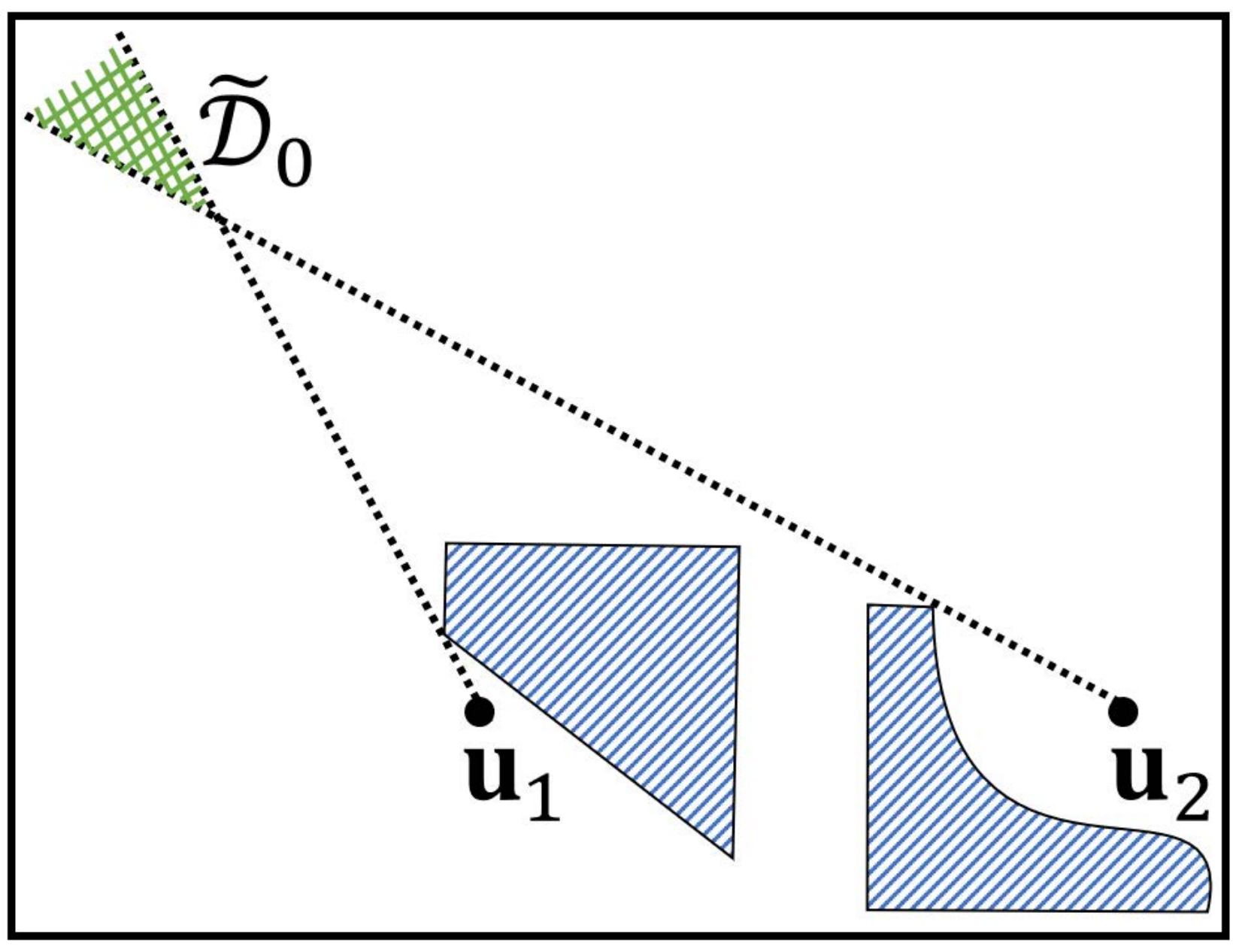}}
\par\end{centering}
\caption{\label{fig:extreme-case} Double-LOS regions sliced at the altitude
$H_{\text{min}}$ under building topologies in extreme cases from
a top view, where the building height is close to the minimum \ac{uav}
altitude $H_{\min}$. (a) $\tilde{\mathcal{D}}_{0}$ appears off the
middle-perpendicular plane between the two users. (b) $\tilde{\mathcal{D}}_{0}$
appears at the top-left in the region, \emph{i.e.}, could be far away
from both users.}
\end{figure}

\section{Algorithm for the Optimal Solution on the Middle-perpendicular Plane}

\label{sec:algorithm_1}

In this section, we solve a simpler version of the problem, where
we aim at finding the optimal \ac{uav} position on the middle perpendicular
plane between the two users. First, two useful properties are investigated
for the \ac{uav} placement problem constrained on the middle perpendicular
plane. Based on these properties, an efficient algorithm is developed.
Then, we prove that the algorithm finds the globally optimal UAV position
on the 2D middle perpendicular plane with a linear trajectory length.

\subsection{Properties on the Middle-perpendicular Plane}

\label{subsec:Search-on-the-middle-perpendicular-plane}

Mathematically, the middle perpendicular plane is specified as $\mathcal{S}=\{\mathbf{p}\in\mathbb{R}^{3}:d_{1}(\mathbf{p})=d_{2}(\mathbf{p})\}$,
which is a 2D plane passing through the midpoint between the two users
at $\mathbf{u}_{1}$ and $\mathbf{u}_{2}$ and perpendicular to the
line connecting the two users. From the definition, one only needs
to focus on minimizing either $d_{1}(\mathbf{p})$ or $d_{2}(\mathbf{p})$.
This property motivates the search on the middle perpendicular plane.

In addition, recall that the two users are assumed to locate at the
ground level, and hence, the middle perpendicular plane is also perpendicular
to the ground. As a consequence, there are two additional properties
summarized in the following lemmas which make it efficient to explore
on the middle perpendicular plane.

The first property is on the double-\ac{los} pattern on the middle
perpendicular plane. Define $\tilde{\mathcal{D}}_{0}^{\text{c}}$
as the set of permissible \ac{uav} positions which are \emph{non-double-LOS}.
\begin{lem}[Double-LOS structure on $\mathcal{S}$]
\label{lem:Double-LOS-structure-on-S} If $\mathbf{p}\in\mathcal{S}\cap\tilde{\mathcal{D}}_{0}$,
then any $\mathbf{p}'\in\mathcal{S}$ perpendicularly above $\mathbf{p}$
also satisfies $\mathbf{p}'\in\mathcal{S}\cap\tilde{\mathcal{D}}_{0}$.
If $\mathbf{p}\in\mathcal{S}\cap\tilde{\mathcal{D}}_{0}^{\text{c}}$,
then any $\mathbf{p}'\in\mathcal{S}$ perpendicularly below $\mathbf{p}$
also satisfies $\mathbf{p}'\in\mathcal{S}\cap\tilde{\mathcal{D}}_{0}^{\text{c}}$.
\end{lem}
\begin{proof}
The first property follows due to the upward invariant property of
$\mathcal{D}_{0}^{(i)}$ and the fact that $\mathcal{S}$ is perpendicular
to the ground. For the second property, assume that $\mathbf{p}'$
is a double-LOS position, \emph{i.e.}, $\mathbf{p}'\in\mathcal{S}\cap\tilde{\mathcal{D}}_{0}$.
Then, according to the upward invariant property, we must have $\mathbf{p}\in\mathcal{S}\cap\tilde{\mathcal{D}}_{0}$,
violating the condition that $\mathbf{p}\in\mathcal{S}\cap\tilde{\mathcal{D}}_{0}^{\text{c}}$,
leading to a contradiction. Therefore, the second property also holds.
\end{proof}
Lemma \ref{lem:Double-LOS-structure-on-S} can be interpreted as follows:
if a position is double-LOS, then all positions perpendicularly above
it are double-LOS; on the other hand, if a position is non-double-LOS,
then all positions perpendicularly below it are non-double-LOS. An
example of the double-LOS pattern on the search plane $\mathcal{S}$
is illustrated in Fig.~\ref{fig:symmetric_trajectory}.

The second property leads to a simplified problem for the search constrained
on the middle perpendicular plane as follows
\begin{equation}
\begin{aligned}\mathscr{P}':\quad\mathop{\mbox{maximize}}\limits _{\mathbf{p}} & \quad F(\mathbf{p})\\
\mathop{\mbox{subject to}} & \quad\mathbf{p}\in\tilde{\mathcal{D}}_{0}\cap\mathcal{S}.
\end{aligned}
\label{sub_problem_middle}
\end{equation}

Define $\mathbf{o}=\frac{1}{2}(\mathbf{u}_{1}+\mathbf{u}_{2})$ as
the midpoint between the two users. Denote the radius from a point
$\mathbf{p}$ on the perpendicular plane $\mathcal{S}$ to the midpoint
$\mathbf{o}$ as
\begin{equation}
r(\mathbf{p})\triangleq||\mathbf{p}-\mathbf{o}||_{2}.\label{eq:radius_rp}
\end{equation}

\begin{lem}[Optimality with minimum radius]
\label{lem:minimize-the-radius} The solution $\hat{\mathbf{p}}$
to $\mathscr{P}'$ minimizes the radius $r(\mathbf{p})$ subject to
$\mathbf{p}\in\tilde{\mathcal{D}}_{0}\cap\mathcal{S}$.
\end{lem}
\begin{proof}
First, note that $d_{1}(\mathbf{p})=d_{2}(\mathbf{p})$ on the middle
perpendicular plane. Second, since $f(d)$ is assumed to be decreasing
in $d$, $F(\mathbf{p})=\text{min}\left\{ f(d_{1}(\mathbf{p})),f(d_{2}(\mathbf{p}))\right\} $
is decreasing in $d_{1}(\mathbf{p})$. Therefore, maximizing $F(\mathbf{p})$
is equivalent to minimizing $d_{1}(\mathbf{p})$. Finally, since $d_{1}(\mathbf{p})=\|\mathbf{p}-\mathbf{u}_{1}\|_{2}=\sqrt{r(\mathbf{p})^{2}+\|\mathbf{u}_{1}-\mathbf{o}\|_{2}^{2}}$,
maximizing $F(\mathbf{p})$ is equivalent to minimizing $r(\mathbf{p})$.
\end{proof}
Lemma~\ref{lem:minimize-the-radius} suggests that the optimal position
on the middle perpendicular plane $\mathcal{S}$ that solves $\mathscr{P}'$
is the closest double-LOS position to the midpoint $\mathbf{o}$.
\begin{figure}
\begin{centering}
\includegraphics[width=0.5\columnwidth]{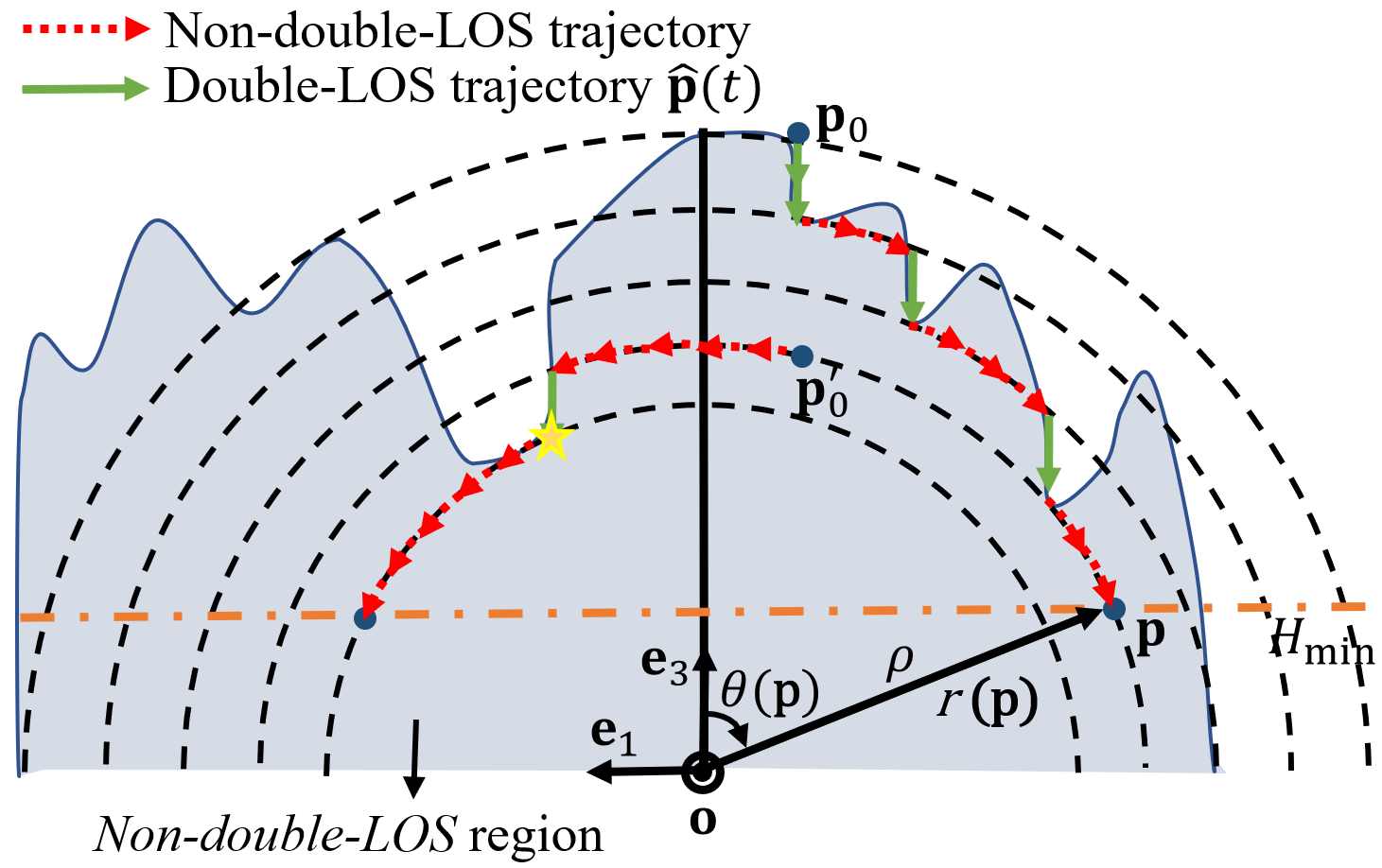}
\par\end{centering}
\caption{Search trajectory on the middle perpendicular plane from the perspective
of $\mathbf{u}_{2}$: flying straight down and along the circle makes
$r(\mathbf{p})$ gradually decrease, and correspondingly, the objective
value $F(\mathbf{p})$ increases.}

\label{fig:symmetric_trajectory}
\end{figure}

\subsection{Search Algorithm}

Using the property of the double-LOS pattern as summarized in Lemma
\ref{lem:Double-LOS-structure-on-S}, the optimal position on the
middle perpendicular plane can be efficiently found following a search
trajectory that starts from a double-LOS initial position and repeats
the following two steps:
\begin{itemize}
\item \textbf{Search downward} whenever the \ac{uav} is at double-LOS;
\item \textbf{Search along the circle} with a fixed radius to the midpoint
$\mathbf{o}$ whenever the \ac{uav} is at non-double-LOS.
\end{itemize}

We now specify the technical details of the above search strategy
using a polar coordinate system $(\rho,\theta)$ defined on the middle
perpendicular plane $\mathcal{S}$. Recall that $\mathbf{e}_{1}$
is a horizontal basis vector perpendicular to $\mathbf{u}_{2}-\mathbf{u}_{1}$,
and $\mathbf{e}_{3}$ is a vertical basis vector pointing upward.
Using the midpoint $\mathbf{o}$ as the origin, the deviation angle
$\theta(\mathbf{p})$ of a position $\mathbf{p}$ \ac{wrt} the direction
$\mathbf{e}_{3}$ shown in Fig.~\ref{fig:symmetric_trajectory} is
computed as
\begin{align}
\theta(\mathbf{p}) & =\text{\text{sign}}(-\mathbf{p}^{\text{T}}\mathbf{e}_{1})\arccos(\mathbf{p}^{\text{T}}\mathbf{e}_{3}/\|\mathbf{p}\|_{2})
\end{align}
where $\text{sign}(-\mathbf{p}^{\text{T}}\mathbf{e}_{1})=1$, if $-\mathbf{p}^{\text{T}}\mathbf{e}_{1}>0$,
indicating $\mathbf{p}$ on the right quadrant, and $\text{sign}(-\mathbf{p}^{\text{T}}\mathbf{e}_{1})=-1$,
if $-\mathbf{p}^{\text{T}}\mathbf{e}_{1}<0$, indicating $\mathbf{p}$
on the left quadrant as shown in Fig.~\ref{fig:symmetric_trajectory}.
As a result, any position $\mathbf{p}$ on the perpendicular plane
$\mathcal{S}$ can be expressed using the polar coordinate $(\rho,\theta(\mathbf{p}))$,
where $\rho=r(\mathbf{p})$ is defined in (\ref{eq:radius_rp}).

Denote the search position at time $t$ as $\mathbf{p}(t)$. Then,
when $\mathbf{p}(t)$ is in non-double-LOS region, the search over
an arc with a fixed radius can be specified by the dynamic equation
$\rho\mathrm{d}\theta=\pm v\mathsf{\mathrm{d}}t$ in the polar coordinate
system $(\rho,\theta)$, where $\rho=r(\mathbf{p}(t))$ and $v$ is
the search speed. The detailed algorithm is summarized in Algorithm~\ref{Alg:symmetric_algorithm},
and an example of search trajectory is shown in Fig.~\ref{fig:symmetric_trajectory}.

Note that Algorithm~\ref{Alg:symmetric_algorithm} requires a double-LOS
initial position $\mathbf{p}_{0}$. Such a position can be found by
increasing the altitude of $\mathbf{p}$ until $\mathbf{p\in\tilde{\mathcal{D}}_{0}}$.
This is because for two outdoor users, double-LOS can be guaranteed
at a high enough altitude for $\mathbf{p}$.
\begin{algorithm}[htbp]
\caption{Dynamic Search Trajectory on the Middle-perpendicular Plane}

\textbf{Input:} Initial double-LOS position $\mathbf{p}_{0}$, and
search speed $v$.

\textbf{Objective: }Design the search trajectory $\mathbf{p}(t)$
and record the double-\ac{los} trajectory $\hat{\mathbf{p}}(t)$.
\begin{enumerate}
\item Initialization: Set $\mathbf{\mathbf{p}}(0)=\mathbf{p}_{0}$ and $\hat{\mathbf{p}}(0)=\mathbf{p}_{0}$.
\item \textbf{\label{enu:Clockwise-search:}Clockwise search:}
\begin{enumerate}
\item \label{enu:Clockwise_if}If {$\mathbf{p}(t)\in\tilde{\mathcal{D}}_{0}$}
then
\begin{enumerate}
\item Set $\hat{\mathbf{p}}(t)=\mathbf{\mathbf{p}}(t)$.
\item Decrease the altitude of $\mathbf{\mathbf{p}}(t)$ according to $\mathrm{d}p_{3}(t)=-v\mathrm{d}t$.
\end{enumerate}
\item \label{enu:Clockwise_else}Else
\begin{enumerate}
\item $\hat{\mathbf{p}}(t)$ remains unchanged, {\em i.e.}, $\mathrm{d}\hat{\mathbf{p}}(t)=0$.
\item \label{enu:Clockwise_else_circle}Move along the circle according
to the dynamical equation: $\rho\textrm{d}\theta=v\textrm{d}t$ expressed
in the polar coordinate system $(\rho,\theta)$, where $\rho=r(\mathbf{p}(t))$.
\end{enumerate}
\item Repeat Step \ref{enu:Clockwise_if} and \ref{enu:Clockwise_else}
until the altitude of $\mathbf{p}(t)$ drops to $H_{\text{min}}$.
\end{enumerate}
\item Define a second initial point $\mathbf{p}_{0}'$ below $\mathbf{p}_{0}$
that satisfies $r(\text{\ensuremath{\mathbf{p}_{0}'}})=r(\hat{\mathbf{p}}(t))$
and $(\mathbf{p}_{0}-\mathbf{p}_{0}')/\|\mathbf{p}-\mathbf{p}_{0}'\|_{2}=\mathbf{e}_{3}$.
Set $\mathbf{p}(t)=\mathbf{p}_{0}'$ and $\hat{\mathbf{p}}(t)$ remains
unchanged.
\item \textbf{Anticlockwise search: }Repeat Step \ref{enu:Clockwise-search:},
but replace the dynamical equation in Step \ref{enu:Clockwise_else_circle}
as $\rho\textrm{d}\theta=-v\textrm{d}t$, until the altitude of $\mathbf{p}(t)$
again drops to $H_{\text{min}}$.
\end{enumerate}
\label{Alg:symmetric_algorithm}
\end{algorithm}

\subsection{Optimality and Complexity of the Search on $\mathcal{S}$}

\label{subsec:analysis}

It turns out that Algorithm~\ref{Alg:symmetric_algorithm} finds
the globally optimal solution to $\mathscr{P}'$ despite that the
double-LOS region $\tilde{\mathcal{D}}_{0}$ can be arbitrarily complicated.
\begin{thm}[Global optimality in 2D]
\label{thm:Algorithm-1-global-optimal-on-S} The double-\ac{los}
trajectory $\hat{\mathbf{p}}(t)$ of Algorithm~\ref{Alg:symmetric_algorithm}
terminates at the globally optimal solution to $\mathscr{P}'$.
\end{thm}
\begin{proof}
See Appendix \ref{sec:Proof-of-Theorem1}.
\end{proof}

Theorem 1 asserts that the global optimality on the 2D middle perpendicular
plane can be guaranteed by a continuous search trajectory which can
be adaptively determined by one of the following two dynamical equations:
$\mathrm{d}p_{3}(t)=-v\textrm{d}t$ and $\rho\textrm{d}\theta=\pm v\textrm{d}t$,
according to the double-\ac{los} status discovered along the trajectory.

In addition, the length of the search trajectory is upper bounded
as shown in the following proposition.
\begin{prop}[Maximum trajectory length]
\label{thm:maximum_trajectory_length} Denote $R_{0}=r(\mathbf{p}_{0})$
as the radius of the initial double-LOS point $\mathbf{p}_{0}$, and
$H_{0}$ is the altitude of $\mathbf{p}_{0}$. The length of the search
trajectory of Algorithm~\ref{Alg:symmetric_algorithm} is upper bounded
by $2(H_{0}-H_{\text{min}})+\pi R_{0}$.
\end{prop}
\begin{proof}
When the \ac{uav} is in double-LOS region, it searches downwards.
The total length of straight down steps is upper bounded by $2(H_{0}-H_{\text{min}})$.
When the \ac{uav} is in non-double-LOS region, it searches along
a circle whose radius is upper bounded by $R_{0}$, and correspondingly,
the total length of these arc-shape steps are upper bounded by $\pi R_{0}$.
Therefore, the upper bound of the total length of the trajectory is
given by $2(H_{0}-H_{\text{min}})+\pi R_{0}$.
\end{proof}
Two observations are made from Theorem~\ref{thm:Algorithm-1-global-optimal-on-S}
and Proposition~\ref{thm:maximum_trajectory_length}. First, to guarantee
a globally optimal solution on the 2D middle perpendicular plane $\mathcal{S}$,
it only requires a search complexity to be a linear function of the
initial distance $R_{0}$ and the initial height $H_{0}$, regardless
of the actual structure of the double-\ac{los} region $\tilde{D}_{0}$.
This is due to the fact that Algorithm~\ref{Alg:symmetric_algorithm}
has exploited the upward invariant property of the double-LOS pattern
as summarized in Lemma \ref{lem:Double-LOS-structure-on-S}.

Second, the fact that the globally optimal solution in 2D is theoretically
guaranteed is also due to the continuous search trajectory where one
needs to determine the double-\ac{los} status for each $\mathbf{p}(t)$
with an infinitesimal step size $\mathrm{d}t$ as described in Algorithm~\ref{Alg:symmetric_algorithm}.
Nevertheless, in a more practical setting in our numerical experiments,
a step size of 5 meters is adopted and the global optimality in 2D
is still numerically observed as shown in Section \ref{sec:Numerical-Results}.

\section{Search for the Optimal Solution in 3D}

In this section, we aim at searching for the globally optimal solution
in 3D for problem $\mathscr{P}$ by exploring a bounded 2D area.

Denote the set of permissible \ac{uav} positions as $\mathcal{P}$.
Denote the critical distance $d_{0}(\mathbf{p})=\max\left\{ d_{1}(\mathbf{p}),d_{2}(\mathbf{p})\right\} $
as the longer distance from the \ac{uav} position $\mathbf{p}$ to
the two users. Given an initial double-LOS point $\mathbf{p}_{0}$
in $\mathcal{P}$, define a region
\[
\mathcal{B}(\mathbf{p}_{0})=\{\text{\ensuremath{\mathbf{p}\in\mathcal{P}}:}d_{0}(\mathbf{p})\text{\ensuremath{\leq}}d_{0}(\mathbf{p}_{0})\}
\]
which geometrically appears as a cap.

It follows that the globally optimal solution $\mathbf{p}^{*}$ to
$\mathscr{P}$ must lie in the cap $\mathcal{B}(\mathbf{p}_{0})$.
To see this, since $f(d)$ is decreasing in $d$, the objective function
$F(\mathbf{p})=\text{min}\left\{ f(d_{1}(\mathbf{p})),f(d_{2}(\mathbf{p}))\right\} $
must also be decreasing in the critical distance $d_{0}(\mathbf{p})$.
Since $\mathbf{p}_{0}$ is a feasible solution and, by definition,
any point $\mathbf{p}\notin\mathcal{B}(\mathbf{p}_{0})$ has a critical
distance $d_{0}(\mathbf{p})$ greater than $d_{0}(\mathbf{p}_{0})$,
implying that the optimal solution cannot be outside $\mathcal{B}(\mathbf{p}_{0})$.

Next, we will narrow down the search area from the 3D cap $\mathcal{B}(\mathbf{p}_{0})$
to bounded 2D areas by algebraically deriving the solution under several
typical \ac{los} patterns.

\subsection{A Compact Search Area on the Perpendicular Plane}

It turns out that it suffices to search on the middle perpendicular
plane $\ensuremath{\mathcal{S}}$ to reveal the \ac{los} status of
the majority part of the cap $\mathcal{B}(\mathbf{p}_{0})$. The key
idea is to map the \ac{los} status from a point in $\mathcal{B}(\mathbf{p}_{0})$
to a point on $\mathcal{S}$ using the colinear invariant property
of the \ac{los} regions $\mathcal{D}_{0}^{(i)}$ for each user $i$.
Specifically, given a point $\mathbf{p}\in\mathcal{B}(\mathbf{p}_{0})$,
to investigate the \ac{los} status for the $i$th user, find a point
$\mathbf{p}_{i}$ on the middle perpendicular plane $\mathcal{S}$,
such that the three points $\mathbf{u}_{i}$, $\mathbf{p}_{i}$, and
$\mathbf{p}$ are colinear as illustrated in Fig~\ref{fig:region_B}.
As a result, according to the colinear invariant property, $\mathbf{p}$
and $\mathbf{p}_{i}$ share the same LOS status for user $i$.
\begin{figure}
\centering{}%
\begin{minipage}[t]{0.49\columnwidth}%
\begin{center}
\includegraphics[width=1\columnwidth]{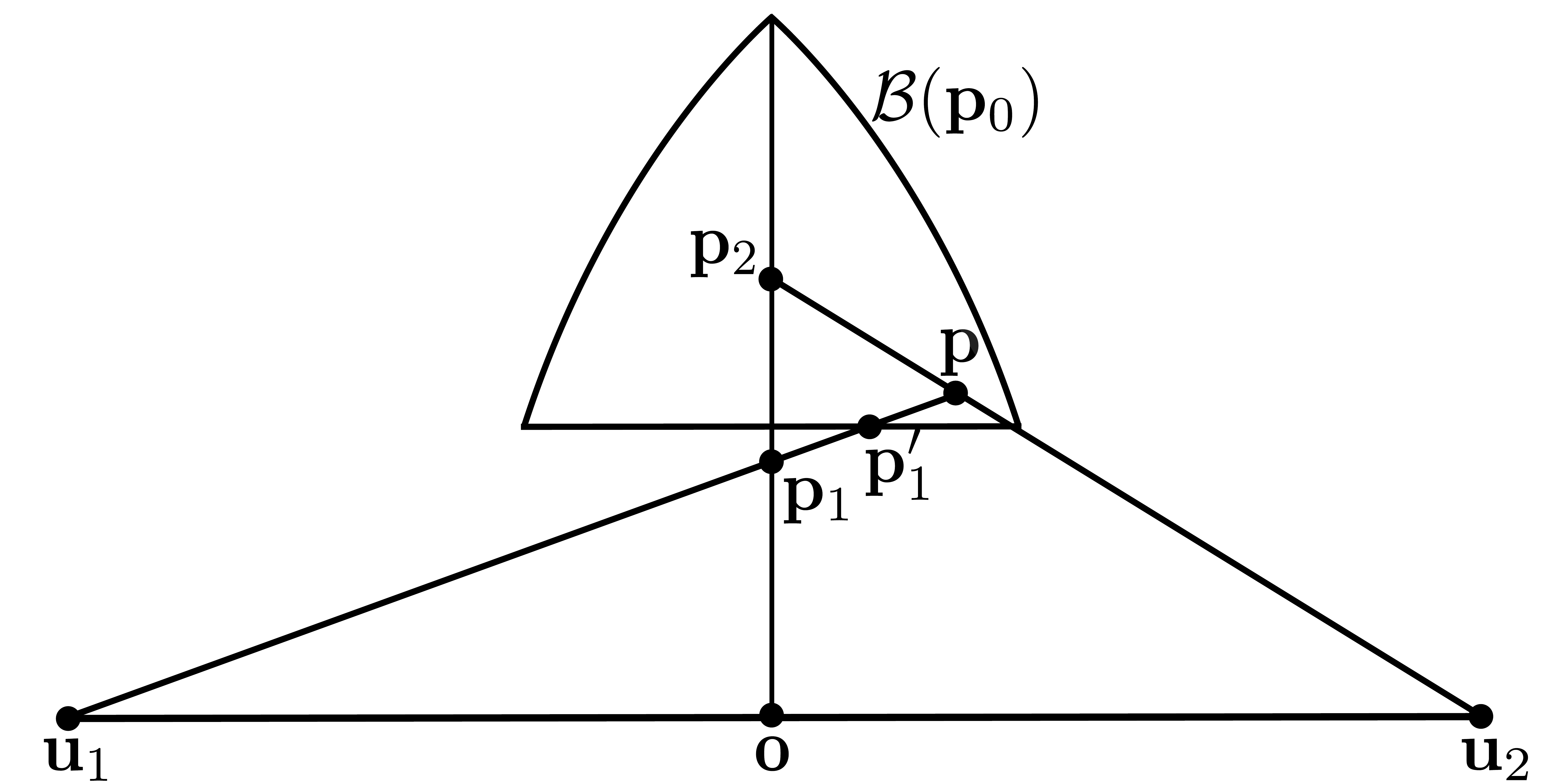}\caption{\label{fig:region_B}Region $\mathcal{B}(\mathbf{p}_{0})$: The \ac{los}
status to $\mathbf{u}_{1}$ of a point $\mathbf{p}\in\mathcal{B}(\mathbf{p}_{0})$
can be determined by that of $\mathbf{p}_{1}\in\mathcal{S}$ or $\mathbf{p}_{1}'\in\mathcal{H}$.}
\par\end{center}%
\end{minipage}%
\begin{minipage}[t]{0.02\columnwidth}%
\includegraphics[width=1\columnwidth]{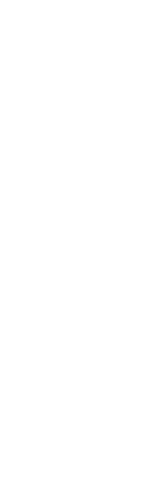}%
\end{minipage}%
\begin{minipage}[t]{0.49\columnwidth}%
\begin{center}
\begin{minipage}[t]{0.5\columnwidth}%
\includegraphics[width=1\columnwidth]{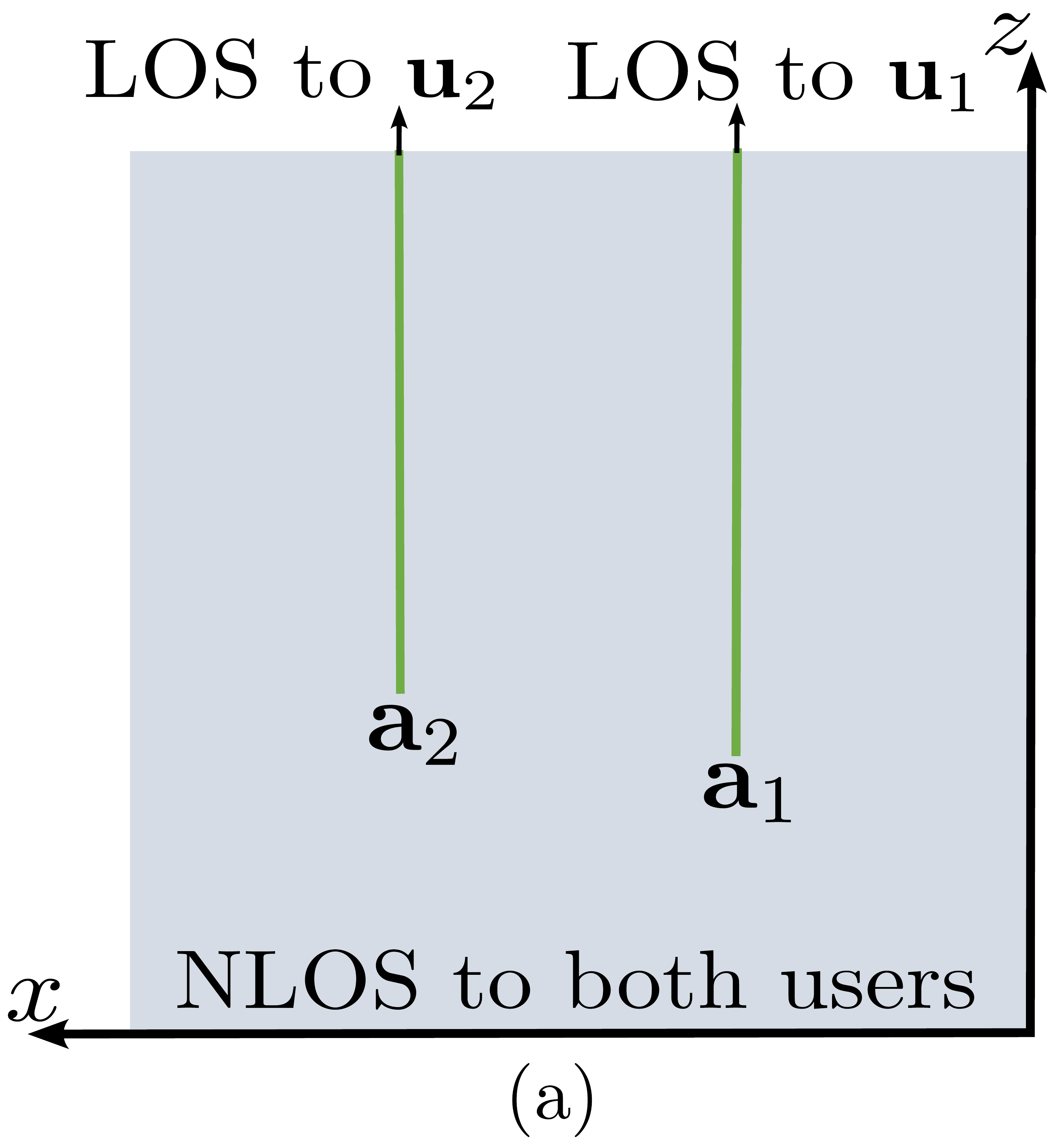}%
\end{minipage}%
\begin{minipage}[t]{0.5\columnwidth}%
\includegraphics[width=1\columnwidth]{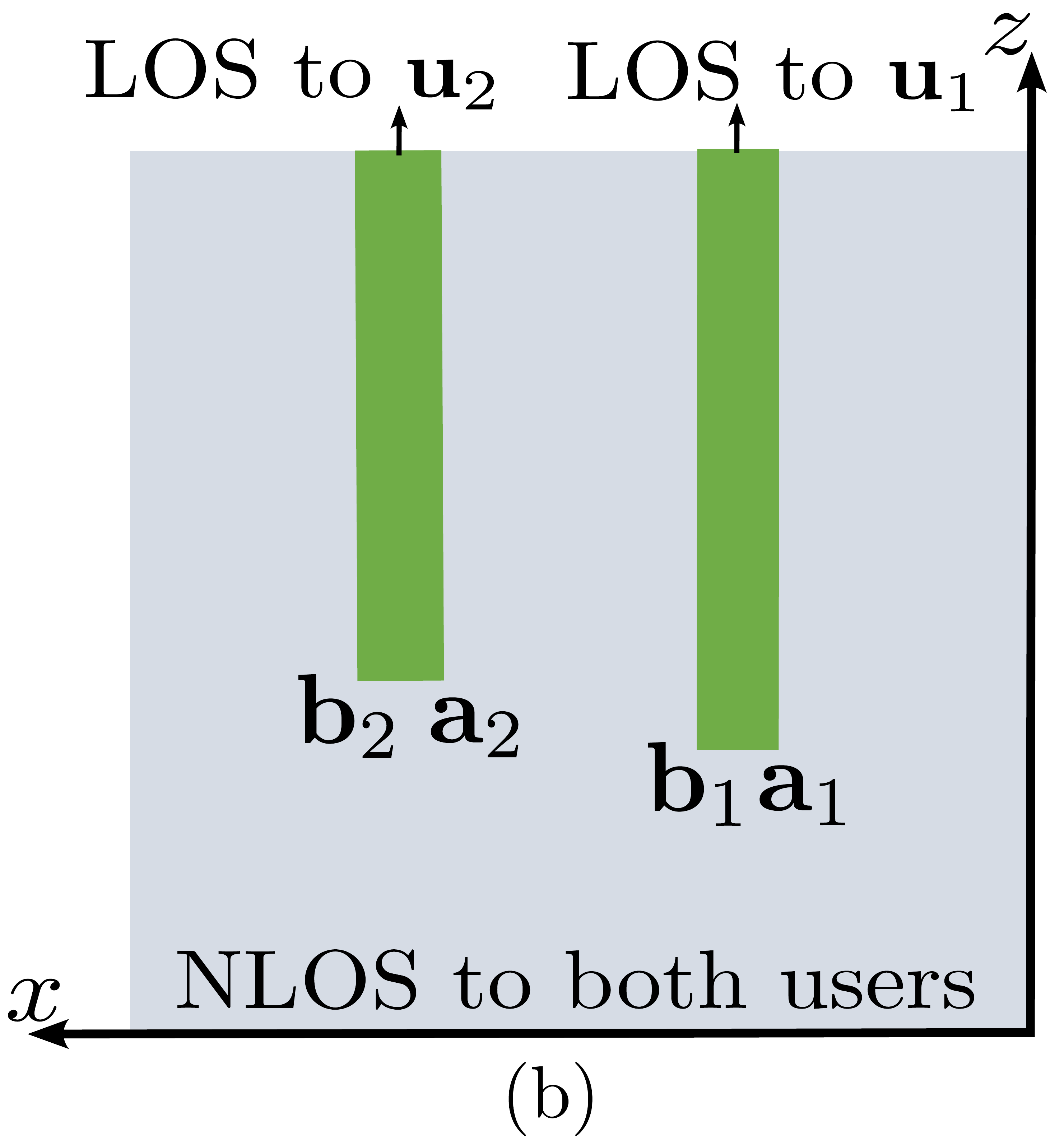}%
\end{minipage}\caption{\label{fig:doube-line-stripe} (a) double-ray LOS pattern; (b) double-stripe
LOS pattern}
\par\end{center}%
\end{minipage}
\end{figure}

However, it is still very challenging to determine the double-\ac{los}
status for $\mathbf{p}$, because one needs to visit two separate
locations, $\mathbf{p}_{1}$ and $\mathbf{p}_{2}$, as shown in Fig.
\ref{fig:region_B}, to determine the \ac{los} status for the two
users, respectively. Thus, the key is how to efficiently combine the
\ac{los} information for the two users along a simple search trajectory.

The first step that we tackle this issue is to develop the closed-form
expression of the globally optimal solution under the simplest \ac{los}
pattern.

1) Double-ray \ac{los} Pattern: Consider two points $\text{\ensuremath{\mathbf{a}_{1}}},\mathbf{a}_{2}\in\mathcal{S}$
on the middle perpendicular plane, where $\mathbf{a}_{1}$ and $\mathbf{a}_{2}$
are LOS positions of user $1$ and user $2$, respectively. According
to the upward invariant property of the LOS regions, the positions
perpendicularly above $\mathbf{a}_{1}$ and $\mathbf{a}_{2}$ are
also \ac{los} for user 1 and user $2$, respectively, as shown in
Fig. \ref{fig:doube-line-stripe}(a).

Suppose that the \ac{los} regions on the middle perpendicular plane
$\mathcal{D}_{0}^{(i)}\cap\mathcal{S}$ are given by the above double-ray
pattern, where the two rays do not overlap. It is clear that there
is no solution on the middle perpendicular plane, but there could
be a solution off the middle perpendicular plane, and the globally
optimal solution $\mathbf{p}^{*}$ to problem $\mathscr{P}$ can be
computed as follows.
\begin{prop}[Double-ray LOS Pattern]
 \label{thm:infer_given_two_positions} Suppose that the set of \ac{los}
positions $\mathcal{D}_{0}^{(1)}\cap\mathcal{S}$ of user $1$ is
a perpendicular ray with a lowest point $\mathbf{a}_{1}=(a_{11},a_{12},a_{13})$.
The set of \ac{los} positions $\mathcal{D}_{0}^{(2)}\cap\mathcal{S}$
of user $2$ is another perpendicular ray with a lowest point $\mathbf{a}_{2}=(a_{21},a_{22},a_{23})$.
If $a_{11}a_{21}>0$, then the globally optimal solution $\mathbf{\mathbf{p}^{*}}$
to $\mathscr{P}$ is given by
\begin{equation}
\mathbf{q}(\mathbf{a}_{1},\mathbf{a}_{2})=\begin{cases}
\frac{2a_{21}}{a_{11}+a_{21}}\mathbf{a}_{1}-\frac{L}{2}\mathbf{e}_{2}, & \text{if \text{ }}\frac{a_{13}}{a_{23}}>\frac{a_{11}}{a_{21}}\\
\frac{2a_{11}}{a_{11}+a_{21}}\mathbf{a}_{2}+\frac{L(a_{21}-2a_{11})}{2a_{11}}\mathbf{e}_{2}, & \text{otherwise}
\end{cases}\label{eq:solution_two_positions}
\end{equation}
where $\mathbf{e}_{2}=(\mathbf{u}_{2}-\mathbf{u}_{1})/\|\mathbf{u}_{2}-\mathbf{u}_{1}\|_{2}$,
and $L=\|\mathbf{u}_{2}-\mathbf{u}_{1}\|_{2}$.
\end{prop}
\begin{proof}
See Appendix \ref{sec:Proof-of-Theorem-double-LOS-line}.
\end{proof}

Naturally, any \ac{los} pattern on the search plane can be modeled
as a union of double-ray \ac{los} patterns parameterized by the endpoints
$\left(\mathbf{a}_{1},\mathbf{a}_{2}\right)$. Thus, we extend the
result to the case of double-stripe \ac{los} pattern as follows.

2) Double-stripe LOS pattern: Consider two \ac{los} vertical regions
with horizontal bottom line segments $\text{\ensuremath{\overline{A_{1}B_{1}}}}$,
and $\overline{A_{2}B_{2}}$. Denote the coordinates of the endpoints
of the line segment $\overline{A_{i}B_{i}}$ as $\mathbf{a}_{i}=(a_{i1},a_{i2},h_{i})$
and $\mathbf{b}_{i}=(b_{i1},b_{i2},h_{i})$, respectively. Without
loss of generality, suppose that $|a_{i1}|\leq|b_{i1}|$, $a_{i1}b_{i1}>0$,
and $h_{i}\geq H_{\text{min}}$ for $i\in\left\{ 1,2\right\} $. Note
that positions above $\text{\ensuremath{\overline{A_{i}B_{i}}}}$
are \ac{los} \ac{wrt} user $i$, according to the upward invariant
property, as shown in Fig.~\ref{fig:doube-line-stripe}(b).

The double-stripe \ac{los} pattern can be constructed as a union
of many double-ray \ac{los} patterns. As a result, if the \ac{los}
regions $\mathcal{D}_{0}^{(i)}\cap\mathcal{S}$ on the middle perpendicular
plane $\mathcal{S}$ appear as a double-stripe \ac{los} pattern,
Proposition \ref{thm:infer_given_two_positions} implies that problem
$\mathscr{P}$ can be equivalently reformulated as
\begin{equation}
\begin{aligned}\mathscr{P}'':\mathop{\mbox{maximize}}\limits _{\mathbf{p}} & \quad F(\mathbf{p})\\
\mathop{\mbox{subject to}} & \quad\mathbf{p}=\mathbf{q}(\mathbf{x}_{1},\mathbf{x}_{2})\\
 & \quad\min\left\{ a_{i1},b_{i1}\right\} \leq x_{i1}\leq\max\left\{ a_{i1},b_{i1}\right\} \\
 & \quad x_{i2}=a_{i2},x_{i3}=h_{i},\text{ for }i=1,2.
\end{aligned}
\label{eq:problem_double-stripe_pattern}
\end{equation}

It is found that problem $\mathscr{P}''$ has a closed-form expression
(\ref{eq:solution_problem-double-stripe}) as derived in Appendix
\ref{sec:The-Closed-form-Solution}. The fact that $\mathscr{P}''$
has a closed-form solution can be understood from the following two
aspects. First, the objective function $F(\mathbf{p})$ is monotonically
decreasing in $d_{0}(\mathbf{p})$, the longer distance from $\mathbf{p}$
to the two users. Thus, the objective is equivalent to minimizing
$d_{0}(\mathbf{p})$, a locally convex function of $\mathbf{p}$ in
the regions of $\left\{ \mathbf{p}\in\mathcal{P}:d_{1}(\mathbf{p})<d_{2}(\mathbf{p})\right\} $
or $\left\{ \mathbf{p}\in\mathcal{P}:d_{1}(\mathbf{p})>d_{2}(\mathbf{p})\right\} $.
Second, from (\ref{eq:solution_two_positions}) and (\ref{eq:problem_double-stripe_pattern}),
the constraint set can be decomposed into a union of several rectangles.
Thus, the intermediate variables $\mathbf{x}_{1}$ and $\mathbf{x}_{2}$,
in $\mathbf{p}=\mathbf{q}(\mathbf{x}_{1},\mathbf{x}_{2})$, must be
found at the endpoints of the intervals $\left(\mathbf{a}_{i},\mathbf{b}_{i}\right)$.
As a result, the closed-form solution (\ref{eq:solution_problem-double-stripe})
is derived via a case-by-case discussion for a total of eight cases.

Denote $Q(\mathbf{a}_{1},\mathbf{b}_{1};\mathbf{\mathbf{a}}_{2},\mathbf{b}_{2})$
as the solution to $\mathscr{P}''$, which is also the solution to
$\mathscr{P}$, under the double-stripe \ac{los} pattern. With the
closed-form solution $Q(\mathbf{a}_{1},\mathbf{b}_{1};\mathbf{\mathbf{a}}_{2},\mathbf{b}_{2})$
to $\mathscr{P}$, one can design simple search trajectories to collect
the endpoints $\left(\mathbf{a}_{i},\mathbf{b}_{i}\right)$, $i\in\left\{ 1,2\right\} $
of \ac{los} segments for both users, to find the best double-\ac{los}
position as will be discussed in Section~\ref{subsec:A-Dynamic-Multi-stage-Algorithm}.

\subsection{A Compact Search Area on the Horizontal Plane}

\label{subsec:Search-on-the-H_min-Horizontal-Plane}

The search on the perpendicular search plane $\mathcal{S}$ has limitations,
because the search height cannot be lower than $H_{\text{min}}$ due
to the problem constraint.\footnote{Recall that, in practice, there could be collision with buildings
if the search altitude of the UAV is not lower bounded.} For a point $\mathbf{p}$ that leads to a colinear point $\mathbf{p}_{1}\in\mathcal{S}$,
\emph{i.e.}, $\mathbf{p}_{1}$, $\mathbf{u}_{1}$, and $\mathbf{p}$
are colinear, with altitude lower than $H_{\text{min}}$ as shown
in Fig.~\ref{fig:region_B}, the point $\mathbf{p}_{1}$ cannot be
reached by the search trajectory on $\mathcal{S}$, and hence, the
\ac{los} status of $\mathbf{p}$ \ac{wrt} $\mathbf{u}_{1}$ cannot
be inferred from $\mathbf{p}_{1}$.

The remedy to such a limitation is to find another colinear point
$\mathbf{p}_{1}'$ on the horizontal plane $\mathcal{H}=\left\{ (x,y,z):z=H_{\text{min}}\right\} $,
such that $\mathbf{u}_{1}$, $\mathbf{p}_{1}$, $\mathbf{p}_{1}'$,
and $\mathbf{p}$ are colinear. Given $\mathbf{p}_{1}'$, the coordinates
of point $\mathbf{p}_{1}$ are calculated as
\begin{equation}
\mathbf{p}_{1}=T_{1}(\mathbf{p}_{1}')\triangleq\frac{L}{2\text{(}p'_{12}+L/2)}(\mathbf{p}_{1}'-\mathbf{u}_{1})+\mathbf{u}_{1}.\label{eq:transition_from_H_S1}
\end{equation}
 Then, according to the colinear invariant property of the \ac{los}
regions, $\mathbf{p}$ is \ac{los} from $\mathbf{u}_{1}$ only if
$\mathbf{p}_{1}$ and $\mathbf{p}_{1}'$ are \ac{los} from $\mathbf{u}_{1}$.
As a result, one can search on $\mathcal{H}$ to discover the \ac{los}
opportunity for $\mathbf{p}$.

Similarly, the \ac{los} status of point $\mathbf{p}_{2}$ \ac{wrt}
$\mathbf{u}_{2}$ can be revealed by the colinear point $\mathbf{p}_{2}'$,
\emph{i.e.},
\begin{equation}
\mathbf{p}_{2}=T_{2}(\mathbf{p}_{2}')\triangleq\frac{L}{2(-p'_{22}+L/2)}(\mathbf{p}_{2}'-\mathbf{u}_{2})+\mathbf{u}_{2}.\label{eq:transition_from_H_S2}
\end{equation}

Combining the two search strategies, we have the following results.
\begin{lem}[Compact 2D Search Areas]
\label{thm:search_area_to_find_optimalp} Given a double-LOS initial
point $\mathbf{p}_{0}$, define $\tilde{\mathcal{B}}(\mathbf{p}_{0})=\mathcal{B}(\mathbf{p}_{0})$,
if $d_{0}(\mathbf{p}_{0})\leq\sqrt{2}L/2$ or $d_{0}(\mathbf{p}_{0})\geq L$,
and $\tilde{\mathcal{B}}(\mathbf{p}_{0})=\{\text{\ensuremath{\mathbf{p}\in\mathcal{P}}:}d_{0}(\mathbf{p})<L^{2}/(2\sqrt{L^{2}-d_{0}^{2}(\mathbf{p}_{0})})\}$,
otherwise. Then, the optimal solution to $\mathscr{P}$ can be found
by searching the \ac{los} points in $\tilde{\mathcal{B}}(\mathbf{p}_{0})\cap\mathcal{S}$
and $\tilde{\mathcal{B}}(\mathbf{p}_{0})\cap\mathcal{H}$.
\end{lem}
\begin{proof}
See Appendix \ref{sec:Proof-of-Theorem-search-area}.
\end{proof}

Lemma \ref{thm:search_area_to_find_optimalp} reveals two important
properties. First, given a double-LOS initial point, the search for
the globally optimal position can be reduced from possibly an 3D unbounded
area to a bounded area $\tilde{\mathcal{B}}(\mathbf{p}_{0})$. Second,
for the 3D globally optimal solution, it suffices to search on a bounded
2D area $\tilde{\mathcal{B}}(\mathbf{p}_{0})\cap\mathcal{S}$ and
$\tilde{\mathcal{B}}(\mathbf{p}_{0})\cap\mathcal{H}$.

\subsection{A Theoretical Bound to the Global Optimality under a Naive Trajectory}

\label{subsec:Gap-to-the-optimality-under-finite-trajectory}

Lemma~\ref{thm:search_area_to_find_optimalp} still requires to search
over a bounded 2D area for the globally optimal solution to $\mathscr{P}$.
However, in practice, the \ac{uav} can only explore through a 1D
trajectory. Therefore, a key question is whether there exists a set
of search trajectories with finite length that guarantee to find a
suboptimal solution, with a performance gap to that of the globally
optimal one upper bounded by a given value.

To answer the above question, we study a set of {\em naive} trajectories
as follows. Denote $H_{0}=\sqrt{d_{0}^{2}(\mathbf{p}_{0})-L^{2}/4}$
as the height at the top of the cap $\mathcal{B}(\mathbf{p}_{0})$,
and $H'_{\text{min}}(\mathbf{p}_{0})=LH_{\text{min}}/(2\sqrt{d_{0}^{2}(\mathbf{p}_{0})-H_{\text{min}}^{2}})$
as the minimum height for $\mathbf{p}_{1}$ to reach the \ac{los}
status for every $\mathbf{p}\in\mathcal{B}(\mathbf{p}_{0})$ as shown
in Fig.~\ref{fig:region_B}. Define an effective search region $\mathcal{B}'(\mathbf{p}_{0})\cap\mathcal{S}$,
where $\mathcal{B}'(\mathbf{p}_{0})\triangleq\{\mathbf{p}\in\mathbb{R}^{3}:p_{3}\geq H'_{\text{min}}(\mathbf{p}_{0}),d_{0}(\mathbf{p})\leq d_{0}(\mathbf{p}_{0})\}$.
Consider a set of equally-spaced search trajectories $\mathcal{\mathcal{T}\in B}'(\mathbf{p}_{0})\cap\mathcal{S}$
parallel to the ground with heights given by $h_{j}=H_{0}-j\delta$,
where $\delta$ is the step size in altitude between adjacent horizontal
trajectories, and $j=1,2,\dots,\text{\ensuremath{\lfloor}}(H_{0}-H'_{\text{min}}(\mathbf{p}_{0}))/\delta\rfloor$.
Note that for $H'_{\text{min}}(\mathbf{p}_{0})\leq h_{j}<H_{\text{min}}$,
the \ac{los} status of $\mathbf{p}_{i}$ on $\mathcal{S}$ needs
to be inferred by searching $\mathbf{p}'_{i}$ on $\mathcal{H}$ as
given in (\ref{eq:transition_from_H_S1}) and (\ref{eq:transition_from_H_S2}).

As a result, the \ac{los} information collected along the trajectory
$\mathcal{T}$ appears as a set of \ac{los} intervals $\mathcal{I}_{1}\triangleq\mathcal{T}\cap\mathcal{D}_{0}^{(1)}$
and $\mathcal{I}_{2}\triangleq\mathcal{T}\cap\mathcal{D}_{0}^{(2)}$,
which are one-dimensional subsets of the \ac{los} regions $\mathcal{D}_{0}^{(1)}$
and $\mathcal{D}_{0}^{(2)}$. Therefore, a suboptimal solution can
be found by solving a set of problems $\mathscr{P}''$ parameterized
by the \ac{los} intervals $\left\{ (\mathbf{a}_{1},\mathbf{b}_{1})\right\} $
collected in $\mathcal{I}_{1}$ and $\left\{ (\mathbf{a}_{2},\mathbf{b}_{2})\right\} $
in $\mathcal{I}_{2}$, and picking the best solution. Mathematically,
this is formulated in the following problem
\begin{equation}
\begin{aligned}\mathop{\mbox{maximize}}\limits _{\mathbf{p}} & \quad F(\mathbf{p})\\
\mathop{\mbox{subject to}} & \quad\mathbf{p}=Q(\mathbf{a}_{1},\mathbf{b}_{1};\mathbf{\mathbf{a}}_{2},\mathbf{b}_{2})\\
 & \quad\left(\mathbf{a}_{1},\mathbf{b}_{1}\right)\in\mathcal{I}_{1}\\
 & \quad\left(\mathbf{a}_{2},\mathbf{b}_{2}\right)\in\mathcal{I}_{2}.
\end{aligned}
\label{eq:problem_found_from_trajectory}
\end{equation}

Let $\tilde{\mathbf{p}}$ be the solution to (\ref{eq:problem_found_from_trajectory}).
Then, if $d_{0}(\tilde{\mathbf{p}})\leq\sqrt{2}L/2$, it is found
that the gap to the globally optimal solution $\mathbf{p}^{*}$ to
$\mathscr{P}$ is bounded linearly in $\delta$, the vertical step
size between adjacent trajectories.
\begin{thm}[Performance Gap to the Globally Optimal Solution]
\label{thm:Alg2-Upper-bound-of-the-performance-gap-1} If the solution
$\tilde{\mathbf{p}}$ to (\ref{eq:problem_found_from_trajectory})
satisfies $d_{0}(\tilde{\mathbf{p}})\text{\ensuremath{\leq}}\sqrt{2}L/2$,
then, the performance gap to the globally optimal solution $\mathbf{p}^{*}$
is upper bounded as 
\begin{equation}
d_{0}(\tilde{\mathbf{p}})-d_{0}(\mathbf{p}^{*})\leq2\delta\sqrt{d_{0}^{2}(\tilde{\mathbf{p}})-H_{\text{min}}^{2}}/L.\label{eq:performance_gap}
\end{equation}
Moreover, if $f(d)$ is convex, then $F(\mathbf{p}^{*})-F(\tilde{\mathbf{p}})\leq-2\delta f'(d_{0}(\mathbf{p}^{*}))\sqrt{d_{0}^{2}(\tilde{\mathbf{p}})-H_{\text{min}}^{2}}/L.$
\end{thm}
\begin{proof}
See Appendix \ref{sec:Proof-of-Theorem-upper-bound-performance-gap-Alg2}.
\end{proof}

Theorem~\ref{thm:Alg2-Upper-bound-of-the-performance-gap-1} finds
that, searching over a set of vertically $\delta$-spaced parallel
trajectories, if the solution computed from (\ref{eq:problem_found_from_trajectory})
satisfies $d_{0}(\tilde{\mathbf{p}})\text{\ensuremath{\leq}}\sqrt{2}L/2$,
then it is guaranteed that the gap to the global optimality is $O(\delta)$.
One can further compute that the coefficient $2\sqrt{d_{0}^{2}(\tilde{\mathbf{p}})-H_{\text{min}}^{2}}/L$
in (\ref{eq:performance_gap}) is upper bounded as 1.4, {\em i.e.},
$d_{0}(\tilde{\mathbf{p}})-d_{0}(\mathbf{p}^{*})\text{\ensuremath{\leq}}1.4\delta$.
In addition, the upper bound of the total length of the search trajectories
in $\mathcal{T}$ can be found upper bounded by $2H_{0}\sqrt{H_{0}^{2}-H_{\text{min}}^{2}}/\delta$
via the total bounded search area divided by the vertical step size
$\delta$. This implies that the total length of the search trajectory
is $O(1/\delta)$.

\subsection{A Dynamic Multi-stage Algorithm\label{subsec:A-Dynamic-Multi-stage-Algorithm}}

Inspiring from Theorem~\ref{thm:Alg2-Upper-bound-of-the-performance-gap-1},
an efficient search trajectory can be developed. Specifically, the
following properties derived from the theoretical bound $O(\delta)$
in Theorem~\ref{thm:Alg2-Upper-bound-of-the-performance-gap-1} can
be exploited. First, a dynamic multi-stage search can be developed,
where one first performs a coarse global search to identify promising
regions, and then, narrow down the search region in subsequent stages
for finer search. Specifically, consider to perform an $M$-stage
search. The first stage follows ($2^{M-1}\delta$)-spaced parallel
trajectories represented by the red lines in Fig.~\ref{fig:Algorithm2_trajectory}.
Then, define the {\em critical trajectories} as the line segments
lying between an \ac{los} segment of a trajectory and an \ac{nlos}
trajectory, for example, the solid green lines in Fig.~\ref{fig:Algorithm2_trajectory}.
According to the upward invariant property, in the next stage, it
suffices to search those critical trajectories for the fine-grained
\ac{los} information. Using such a search strategy, it can be found
that the total length of the critical trajectories in an $M$-stage
search is upper bounded as $2H_{0}\sqrt{H_{0}^{2}-H_{\text{min}}^{2}}/(2^{M-1}\delta)+2(M-1)\sqrt{H_{0}^{2}-H_{\text{min}}^{2}}$,
which is minimized at $M=g_{W}(H_{0}\ln2/\delta)/\ln2$, where $g_{W}(x)$
is the Lambert W function of $x$. We find that in a typical setting,
{\em e.g.}, $H_{0}\in[120,160]$ meters and $\delta\in[2,3]$ meters,
the optimal $M$ is $4$.

Second, the $O(\delta)$ bound helps filter out unpromising region
for subsequent finer search using the coarse \ac{los} information
obtained at the previous stages. For example, if $\tilde{\mathbf{p}}_{1}$
and $\tilde{\mathbf{p}}_{2}$ are both found as double-LOS points
with $d_{0}(\tilde{\mathbf{p}}_{2})-d_{0}(\tilde{\mathbf{p}}_{1})>2^{M-m+1}\delta\sqrt{d_{0}^{2}(\tilde{\mathbf{p}}_{2})-H_{\text{min}}^{2}}/L$,
then there is no need to finely search the local area related to $\tilde{\mathbf{p}}_{2}$,
because of the $O(\delta)$ upper bound in Theorem~\ref{thm:Alg2-Upper-bound-of-the-performance-gap-1}.
For example, the dashed green critical trajectory in Fig.~\ref{fig:Algorithm2_trajectory}
may be ignored if it is found substantially less promising than the
other critical trajectories in solid green.

Finally, there exist path planning algorithms to connect the isolated
critical trajectories using a short path. The overall search strategy
is summarized in Algorithm~\ref{Alg:symmetric_algorithm_2}, and
an typical realization of the trajectory is demonstrated in Fig.~\ref{fig:Algorithm2_trajectory},
where the target search region is reduced stage-by-stage.
\begin{figure}
\begin{centering}
\includegraphics[width=0.5\columnwidth]{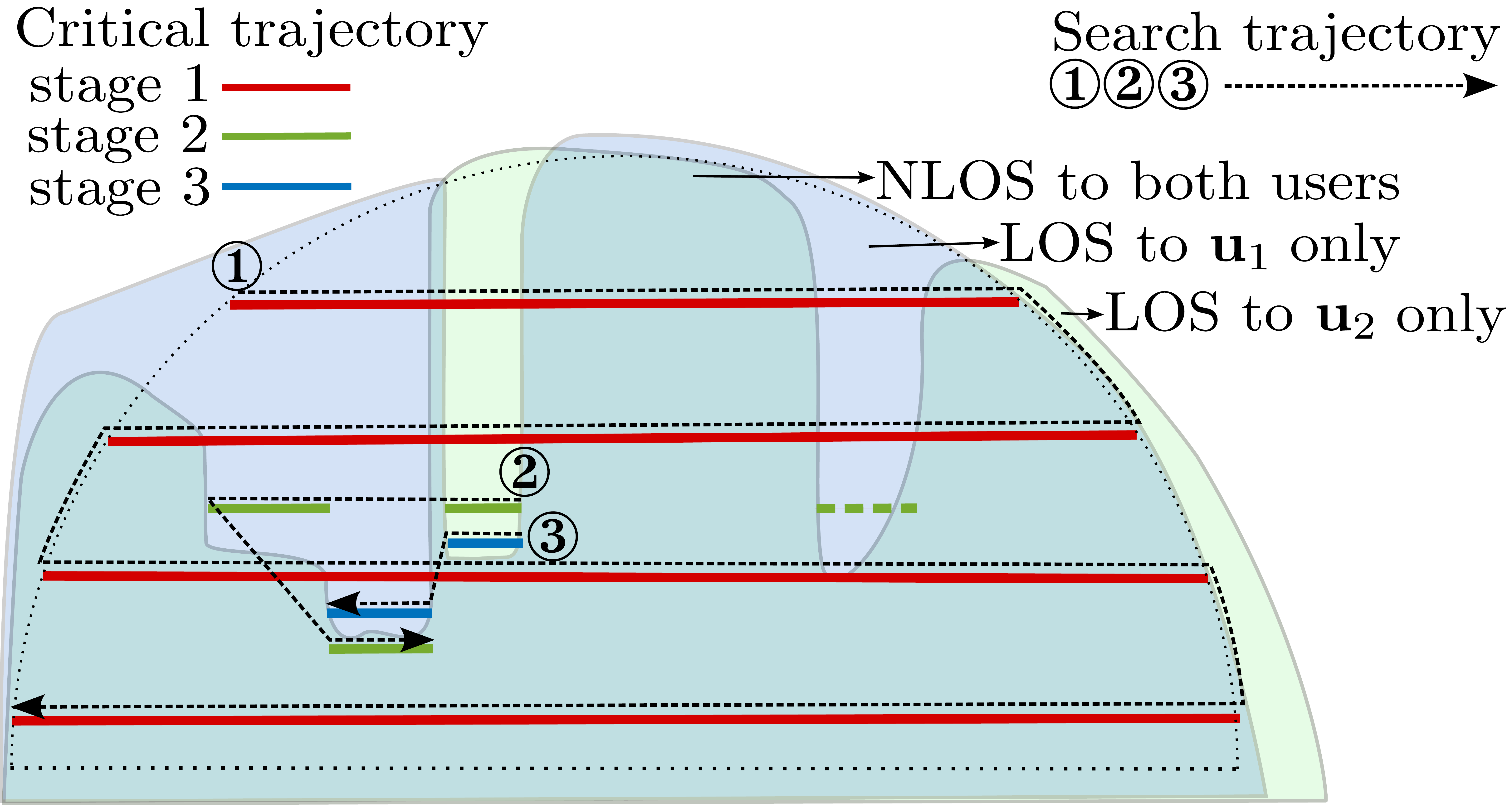}
\par\end{centering}
\caption{\label{fig:Algorithm2_trajectory}An example of search trajectory
on $\mathcal{B}(\mathbf{p}_{0})\cap\mathcal{S}$: The bold solid line
represents the critical trajectory. The thin dashed lines with arrows
indicate a search trajectory to connect the critical trajectories
at different stages ($m=1,2,3$). The light green and light blue shaded
areas portray the NLOS patterns to user 1 and user 2, respectively.
The overlapping region of the two NLOS patterns is NLOS to both users.}
\end{figure}
\begin{algorithm}
\begin{enumerate}
\item Initialization: Based on a double-LOS initial point $\mathbf{p}_{0}$,
set the initial solution $\tilde{\mathbf{p}}(1)=\mathbf{p}_{0}$.
Initialize $\mathcal{I}{}_{1}=\varnothing$, $\mathcal{I}{}_{2}=\varnothing$,
$H_{0}=\sqrt{d_{0}^{2}(\mathbf{p}_{0})-L^{2}/4}$, $H'_{\text{min}}(\tilde{\mathbf{p}}(1))=LH_{\text{min}}/(2\sqrt{d_{0}^{2}(\tilde{\mathbf{p}}(1))-H_{\text{min}}^{2}})$,
$j=1$, and $m=1$.
\item \label{enu:M-stage}Set $h_{j}=H_{0}-j2^{M-m}\delta$, for a step
size $\delta$. Initialize $\mathcal{I}_{i}^{(j)}=\varnothing$.
\item \textbf{Configure the initial search trajectory.}
\begin{enumerate}
\item \label{enu:If-:-Search_above_H_min}If $h_{j}\geq H_{\text{min}}$:
Search on $\mathcal{S}$ with height $h_{j}$ between $x=\pm\sqrt{d_{0}^{2}(\tilde{\mathbf{p}}(m))-L^{2}/4-h_{j}^{2}}$.
For any \ac{los} segment discovered for user $i$ with endpoints
$\mathbf{a}_{i}$ and $\mathbf{b}_{i}$, assign $(\mathbf{a}_{i},\mathbf{b}_{i})\rightarrow\mathcal{I}_{i}^{(j)}$.
\item \label{enu:If-:-Search_below_H_min}If $h_{j}<H_{\text{min}}$: Search
on $\mathcal{H}$ with $y=L/2\pm(H_{\text{min}}/h_{j}-1)L/2$ between
$x=\pm\sqrt{d_{0}^{2}(\tilde{\mathbf{p}}(m))-H_{\text{min}}^{2}-L^{2}H_{\text{min}}^{2}/(4h_{j}^{2})}$.
For any \ac{los} segment discovered for user $i$ with endpoints
$\mathbf{a}_{i}$ and $\mathbf{b}_{i}$, assign $(T_{i}(\mathbf{a}_{i}),T_{i}(\mathbf{b}_{i}))\rightarrow\mathcal{I}_{i}^{(j)}$
where $T_{i}$ for $i=1,2$ are given in (\ref{eq:transition_from_H_S1})
and (\ref{eq:transition_from_H_S2}).
\item $\mathcal{I}_{i}=\mathcal{I}_{i}\cup\mathbf{\mathcal{I}}_{i}^{(j)}$,
$j\leftarrow j+1$, Repeat from Step~\ref{enu:M-stage} until $h_{j}<H'_{\text{min}}(\tilde{\mathbf{p}}(m))$.
\end{enumerate}
\item Calculate $\tilde{\mathbf{p}}(m)$ as the solution to problem (\ref{eq:problem_found_from_trajectory}).
\item For the $k$th interval $(\mathbf{a}_{i}^{(k)},\mathbf{b}_{i}^{(k)})$
in $\mathcal{I}_{i}$ where $k\in\{1,2,\dots,|\mathcal{I}_{i}|\}$
and $i\in\{1,2\}$, calculate the solution $\text{\ensuremath{\tilde{\mathbf{p}}_{i}^{(k)}}}$
to problem (\ref{eq:problem_found_from_trajectory}) by replacing
$\mathcal{I}_{i}$ with $\left\{ (\mathbf{a}_{i}^{(k)},\mathbf{b}_{i}^{(k)})\right\} $.
If $d_{0}(\tilde{\mathbf{p}}_{i}^{(k)})-d_{0}(\tilde{\mathbf{p}}(m))>2^{M-m+1}\delta\sqrt{d_{0}^{2}(\tilde{\mathbf{p}}_{i}^{(k)})-H_{\text{min}}^{2}}/L$,
then remove $(\mathbf{a}_{i}^{(k)},\mathbf{b}_{i}^{(k)})$ from $\mathcal{I}_{i}$.
\item Reset $j=1$, and update $m\leftarrow m+1$.
\item \textbf{Configure the refined search trajectory.}
\begin{enumerate}
\item \label{enu:For-each-interval_search}For each interval $(\mathbf{a}_{i},\mathbf{b}_{i})$
in $\mathcal{I}_{i}$, set $H_{0}=a_{i3}$, and repeat from Step~\ref{enu:M-stage},
but change the range of $x$ in Step~\ref{enu:If-:-Search_above_H_min}
as $[a_{i1},b_{i1}]$ and in Step~\ref{enu:If-:-Search_below_H_min}
as $[a_{i1}H_{\text{min}}/a_{i3},b_{i1}H_{\text{min}}/b_{i3}]$, until
$h_{j}<H_{0}-2^{M-m}\delta$.
\item Repeat Step~\ref{enu:For-each-interval_search} until $m=M$, and
then, output $\tilde{\mathbf{p}}(M)$.
\end{enumerate}
\end{enumerate}
\caption{$M$-stage Dynamic Search for 3D Optimal UAV Placement}

\label{Alg:symmetric_algorithm_2}
\end{algorithm}

\section{Numerical Results}

\label{sec:Numerical-Results}

\selectlanguage{english}%
In this section, the proposed algorithms are compared with four baseline
schemes on four real-world city maps.
\selectlanguage{american}%

\subsection{Environment Setup and Scenarios}

\label{subsec:Propagation-Environment-Modeling}

We perform experiments over four city topologies from real data. As
shown in Fig.~\ref{fig:Four-local-areas}, Map A and Map B are 3D
maps of two different areas in Beijing, China. They represent the
3D environment of typical commercial center and traditional commercial
area, respectively. Map C and Map D are 2D street maps of two different
areas in Guangzhou, China. Based on the street maps, we manually generate
the height of the buildings following a uniform distribution of $[50,80]$
meters. The two simulated environments respectively represent the
modern dense residential area and the ultra dense area probably appear
in the future.

The characteristics of the four maps are summarized in Table~\ref{tab:Maps-with-four-function-areas},
where we use the building coverage ratio (BCR) and floor area ratio
(FAR) \cite{GonCreHerRod:J13} to quantify the building density of
the areas. It is observed that Map A is the most sparse, and Map D
has the largest building density.

For the experiment on each map, the minimum height of the \ac{uav}
$H_{\text{min}}$ is set as the maximum building height in the map,
as shown in Table~\ref{tab:Maps-with-four-function-areas}, to avoid
potential collision of the \ac{uav}. There are $5,000$ user pairs
placed uniformly at random in the non-building area of each map.
\begin{figure}
\begin{centering}
\subfigure[Map A: Typical commercial center]{\includegraphics[width=0.29\columnwidth]{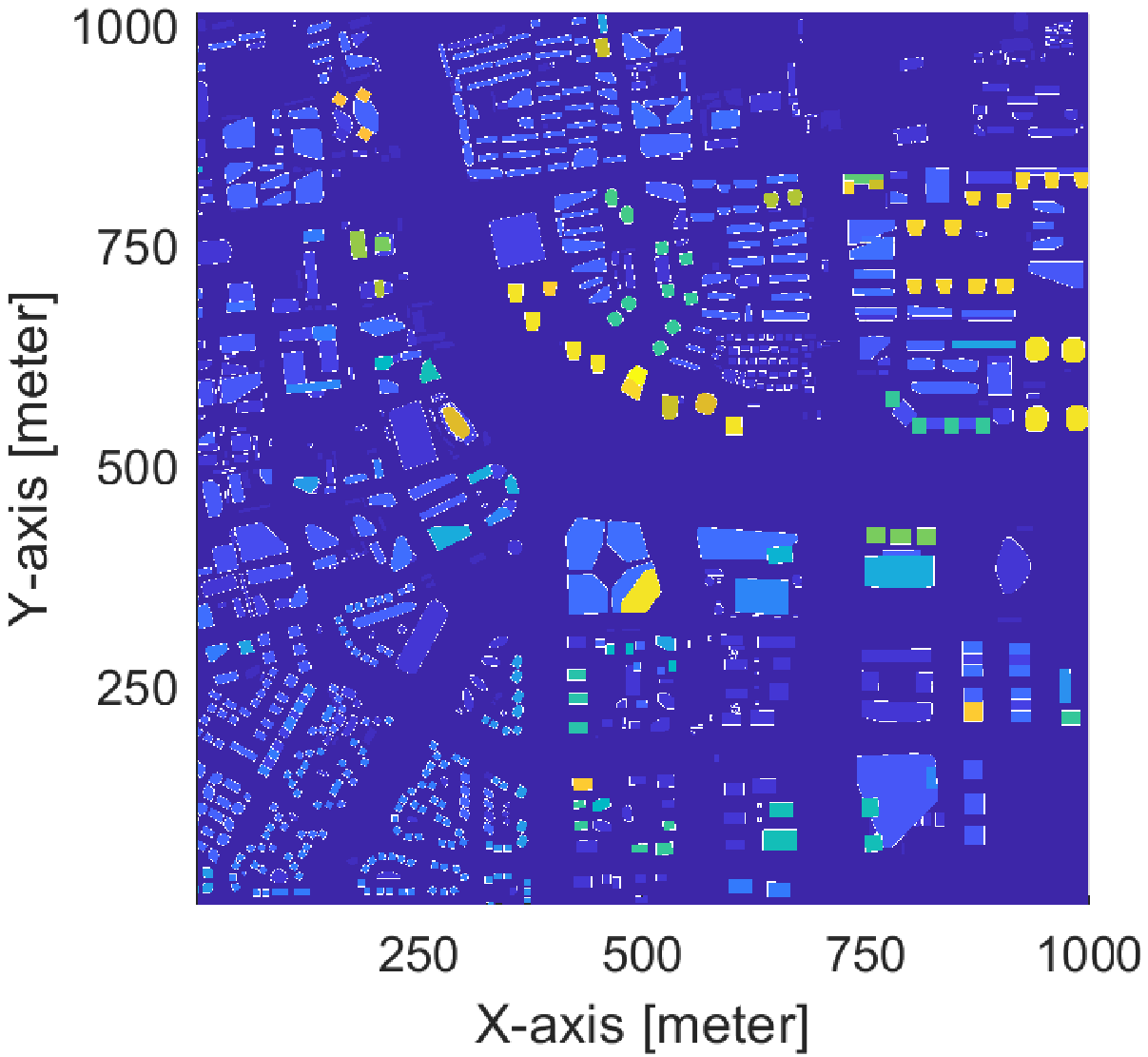}}\subfigure[Map B: Traditional commercial area]{\includegraphics[width=0.32\columnwidth]{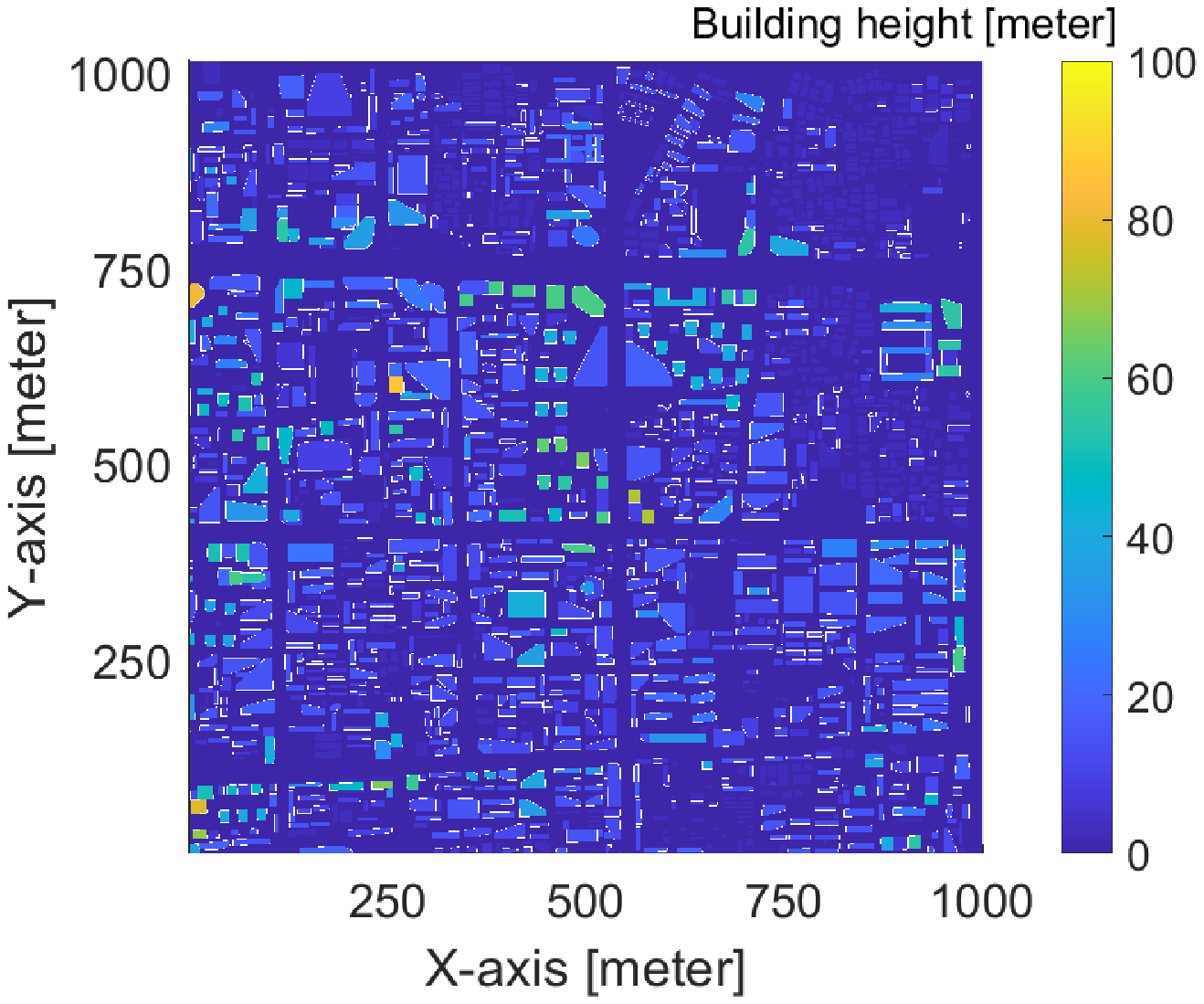}}
\par\end{centering}
\begin{centering}
\subfigure[Map C: Dense residential area]{\includegraphics[width=0.32\columnwidth]{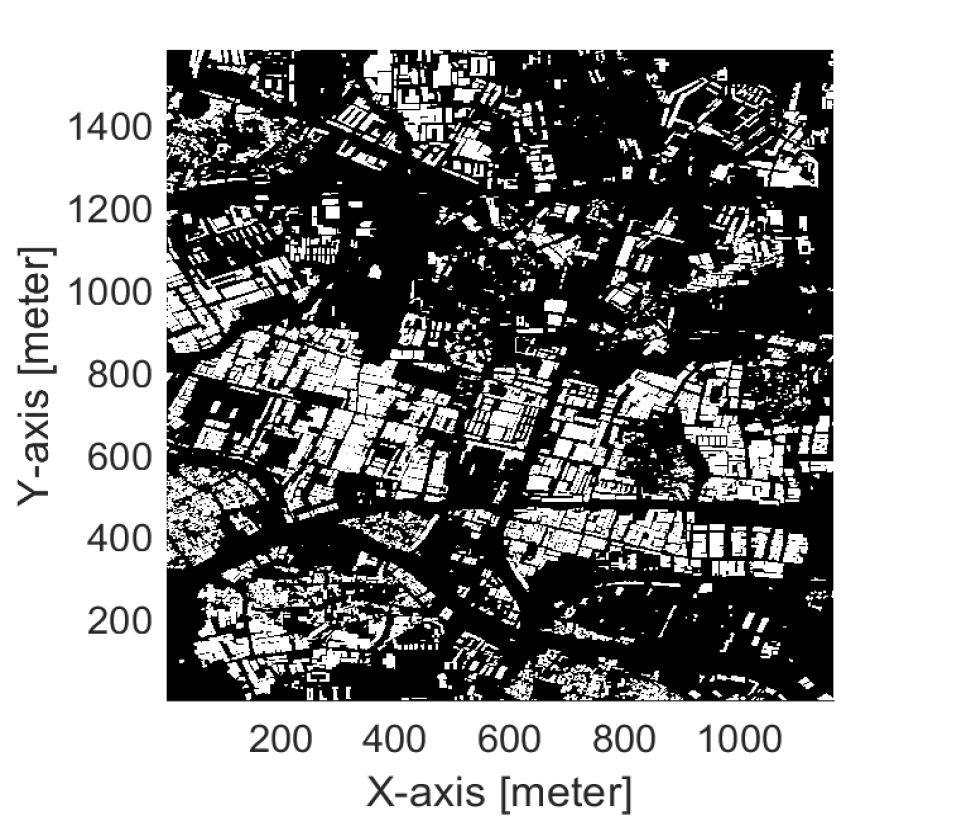}}\subfigure[Map D: Ultra dense area]{\includegraphics[width=0.28\columnwidth]{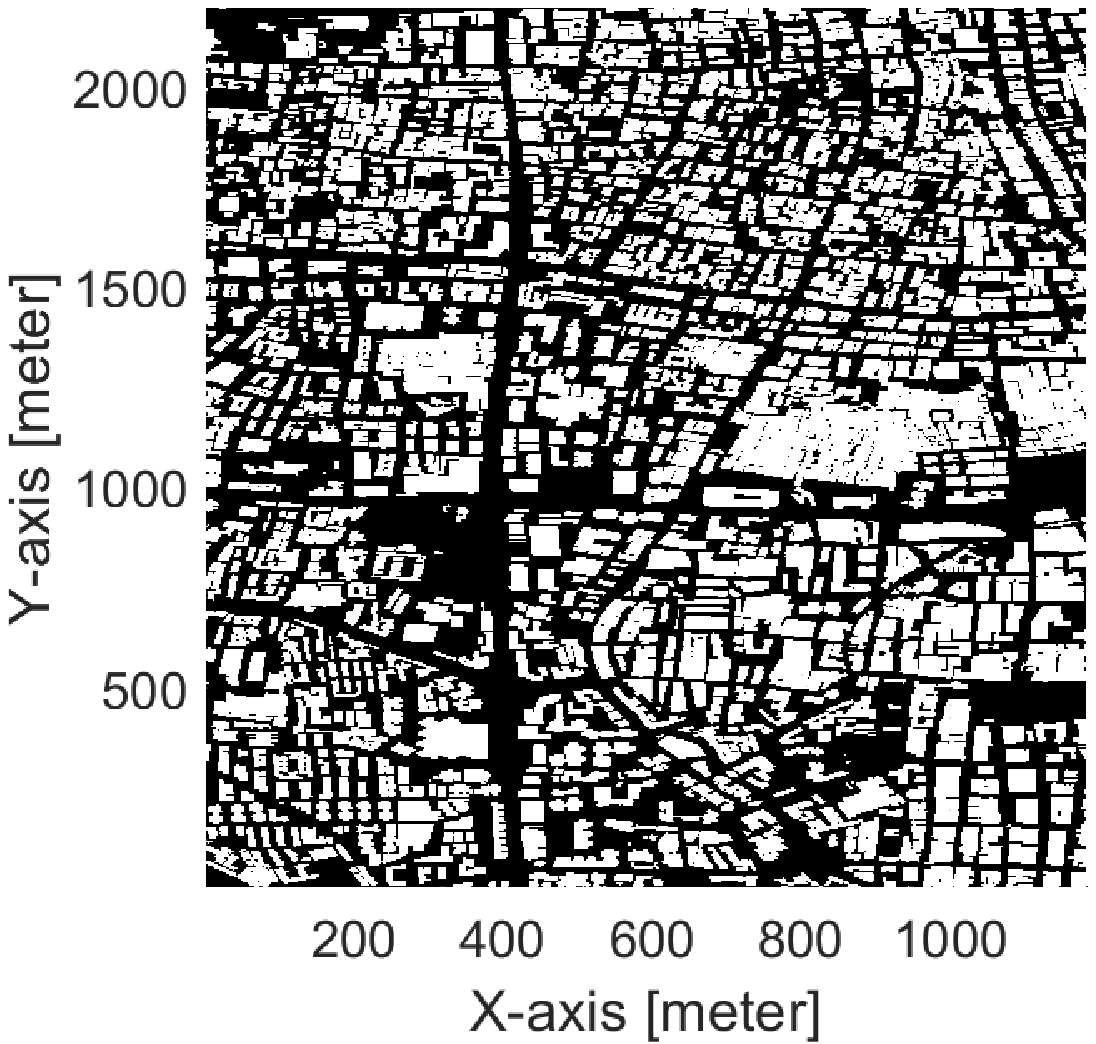}}
\par\end{centering}
\caption{\label{fig:Four-local-areas}Four local areas, where Map A and B are
3D city maps of different areas in Beijing, China, and Map C and D
are generated from 2D street maps from Guangzhou, China, with simulated
building heights.}

\label{fig:simulation_maps-2}
\end{figure}

\selectlanguage{english}%
Two application scenarios are evaluated in our experiments.
\begin{itemize}
\item {\bf UAV relaying:} A UAV is placed to establish \ac{los} relay
channels for two ground users under decode-and-forward relaying. Consider
the path loss model of millimeter wave cellular reported in \cite{CheMitGes:T21}
as $\text{PL}_{\text{LOS}}(d)=61.4+20.0\log_{10}(d)$ with the shadowing
parameter $\sigma_{\text{sf}}=1$~dB for \ac{los} link, and $\text{\ensuremath{\text{PL}_{\text{NLOS}}}}(d)=72.0+29.2\log_{10}(d)$
with $\sigma_{\text{sf}}=5$~dB for NLOS link at a carrier frequency
of $28$~GHz. Correspondingly, the performance evaluation function
is defined as the channel capacity $f(d)=W\log(1+P\cdot10^{(-\text{PL}(d)-\sigma_{\text{sf}})/10}/(WN_{0}))$
where $W=1$ GHz is the allocated bandwidth, $P$ is the transmission
power, and $N_{0}$ is the noise power spectrum density set as $-169$
dBm/Hz.
\item {\bf UAV WPT:} A \ac{uav} is placed to wirelessly charge two ground
users simultaneously. The evaluation function $f(d)$ of the WPT channel
is adopted from the linear harvesting model as $f(d)=\eta P\beta/d^{\alpha}$
where $\eta=60\%$ denotes the linear RF-to-direct current (DC) energy
conversion efficiency, $P=40$ dBm denotes the transmit power, $\beta=-30$
dB denotes the channel power gain at reference distance $d_{0}=1$
meter, and $\alpha=3$ denotes the path loss exponent \cite{XieCaoXuZha:T21}.
\end{itemize}
\foreignlanguage{american}{}
\begin{table}
\selectlanguage{american}%
\caption{\label{tab:Maps-with-four-function-areas}Maps with four types of
function areas}

\renewcommand{\arraystretch}{1.2}
\centering{}%
\begin{tabular}{>{\raggedright}p{0.165\columnwidth}>{\raggedright}p{0.04\columnwidth}>{\raggedright}p{0.04\columnwidth}>{\raggedright}p{0.175\columnwidth}>{\raggedright}p{0.21\columnwidth}>{\raggedright}p{0.22\columnwidth}}
\hline 
\textbf{Map} & \centering{}\textbf{BCR} & \centering{}\textbf{FAR} & \centering{}\textbf{Mean height {[}meter{]}} & \centering{}\textbf{Maximum height {[}meter{]}} & \textbf{Comment}\tabularnewline
\hline 
Map A (Beijing) & \centering{}19\% & \centering{}1.4 & \centering{}22 & \centering{}96 & Typical commercial center\tabularnewline
Map B (Beijing) & \centering{}32\% & \centering{}1.8 & \centering{}16 & \centering{}87 & Traditional commercial area\tabularnewline
Map C (Guangzhou) & \centering{}22\% & \centering{}4.9 & \centering{}65 & \centering{}80 & Dense residential area\tabularnewline
Map D (Guangzhou) & \centering{}40\% & \centering{}8.8 & \centering{}65 & \centering{}80 & Ultra dense area\tabularnewline
\hline 
\end{tabular}\selectlanguage{english}%
\end{table}

\selectlanguage{american}%
We evaluate the following baseline schemes for performance benchmarking.
The exhaustive search schemes are implemented using a 5-meter step
size.
\begin{itemize}
\item \textit{Exhaustive 3D search}: This scheme performs an exhaustive
search in the 3D search space above the area of interest.
\item \textit{Exhaustive 2D searc}\emph{h (horizontal)} \cite{LyuZenZhaLim:J17}:
This scheme performs an exhaustive search over a 2D horizontal plane
$\{\mathbf{p}\in\mathbb{R}^{3}:p_{3}=H_{\text{2D}}\}$ where $H_{\text{2D}}$
is set as 120 meters here.
\item \emph{Exhaustive 2D search (vertical)}: This scheme is designed to
confirm the optimality of the output of Algorithm~1 on the middle
perpendicular plane $\mathcal{S}$. It performs an exhaustive search
over the 2D middle perpendicular plane $\mathcal{S}$.
\item \textit{Statistical }\emph{method} \cite{AlhKanLar:J14,KumSinDarSha:J21,CheHua:J22}:
The average path loss from the \ac{uav} position $\mathbf{p}$ to
the $i$th ground user can be formulated as
\begin{equation}
\text{PL}_{\text{ave}}=\text{P}_{\text{LOS}}(\mathbf{p})\times\text{PL}_{\text{LOS}}(d_{i}(\mathbf{p}))+(1-\text{P}_{\text{LOS}}(\mathbf{p}))\text{PL}_{\text{NLOS}}(d_{i}(\mathbf{p}))
\end{equation}
where $\text{P}_{\text{LOS}}(\mathbf{p})$ is the \ac{los} probability
of the \ac{uav} position $\mathbf{p}$. The \ac{los} probability
$\text{P}_{\text{LOS}}(\mathbf{p})$ is defined as
\begin{equation}
\text{P}_{\text{LOS}}(\mathbf{p})=\frac{1}{1+a\times\text{exp}(-b(\text{arctan}(p_{3}/r_{i})-a))}
\end{equation}
where $r_{i}=\sqrt{\|\mathbf{p}-\mathbf{u}_{i}\|_{2}^{2}-p_{3}^{2}}$,
and the environmental parameter pair $(a,b)$ is learned from the
actual distribution of \ac{los} regions, and obtained as $(2.60,0.05)$,
$(58.91,8.90),$ $(63.77,3.95)$, and $(64.33,141.17)$ in Maps A,
B, C, and D, respectively.
\end{itemize}

\subsection{UAV Communication Relaying}

Fig. \ref{fig:Capacity-on-diffenert-function-areas} summarizes the
mean capacity of different schemes on the four maps under transmission
power $P=30$ dBm. The numerical result in Fig. \ref{fig:Capacity-on-diffenert-function-areas}
confirms the global optimality of Algorithm~1 on the perpendicular
plane $\mathcal{S}$. The mean capacity of Algorithm~1 and exhaustive
2D search on $\mathcal{S}$ is the same as each other in the four
function areas. It is worth noting that the trajectory length of Algorithm~1
is at most $1/45$ of that of the Exhaustive 2D search on $\mathcal{S}$
as shown in Table \ref{tab:comparison-of-mean-trajectory-length}.
Thus, Algorithm~1 is much more efficient than the Exhaustive 2D search
on $\mathcal{S}$.

In Map A and B, where the distribution of buildings is relatively
sparse, both of the two proposed algorithms achieve mostly the globally
optimal performance with a negligible performance gap to the Exhaustive
3D scheme as shown in Fig. \ref{fig:Capacity-on-diffenert-function-areas},
although Algorithm~1 only searches on the middle perpendicular plane
with limited search length. Such a result suggests that the globally
optimal solution has a high chance to locate on the middle perpendicular
plane over a sparse city topology.

In Map C and D, where buildings are denser, the performance of Algorithm~1
degrades and it achieves only $85.4\%$ on Map C and $75.2\%$ on
Map D to the Exhaustive 3D scheme in Fig.~\ref{fig:Capacity-on-diffenert-function-areas}.
By contrast, the performance of Algorithm~2 with $\delta=3$ meters
is still close, {\em i.e.}, above $99.8\%$ in Map C and above $98.0\%$
in Map D, to that of the Exhaustive 3D scheme. This is because Algorithm~\ref{Alg:symmetric_algorithm_2}
is capable of discovering those potentially better \ac{los} positions
off the middle perpendicular plane using coarse \ac{los} information
on a bounded 2D region. However, the performance of the statistical
method is relatively poor since this method does not examine the actual
obstacle occlusion, resulting in no \ac{los} guarantee in practical
applications.

Fig.~\ref{fig:Capacity-under-different-power} illustrates the average
capacity versus the transmission power. For a sparse city topology
(Map A), both of the proposed algorithms achieve almost identical
performance to the Exhaustive 3D scheme. For a dense topology (Map
C), Algorithm~1 degrades from the Exhaustive 3D scheme, but the performance
degradation is small at the high transmission power regime. Specifically,
Algorithm~1 achieves above $90\%$ to the Exhaustive 3D scheme under
transmission power of $40$~dBm. In contrast, Algorithm~2 achieves
above $99.7\%$ to the Exhaustive 3D scheme at different transmission
powers under $\delta=3$ meters.
\begin{figure}
\centering{}%
\begin{minipage}[t]{0.5\columnwidth}%
\begin{center}
\includegraphics[width=1\columnwidth]{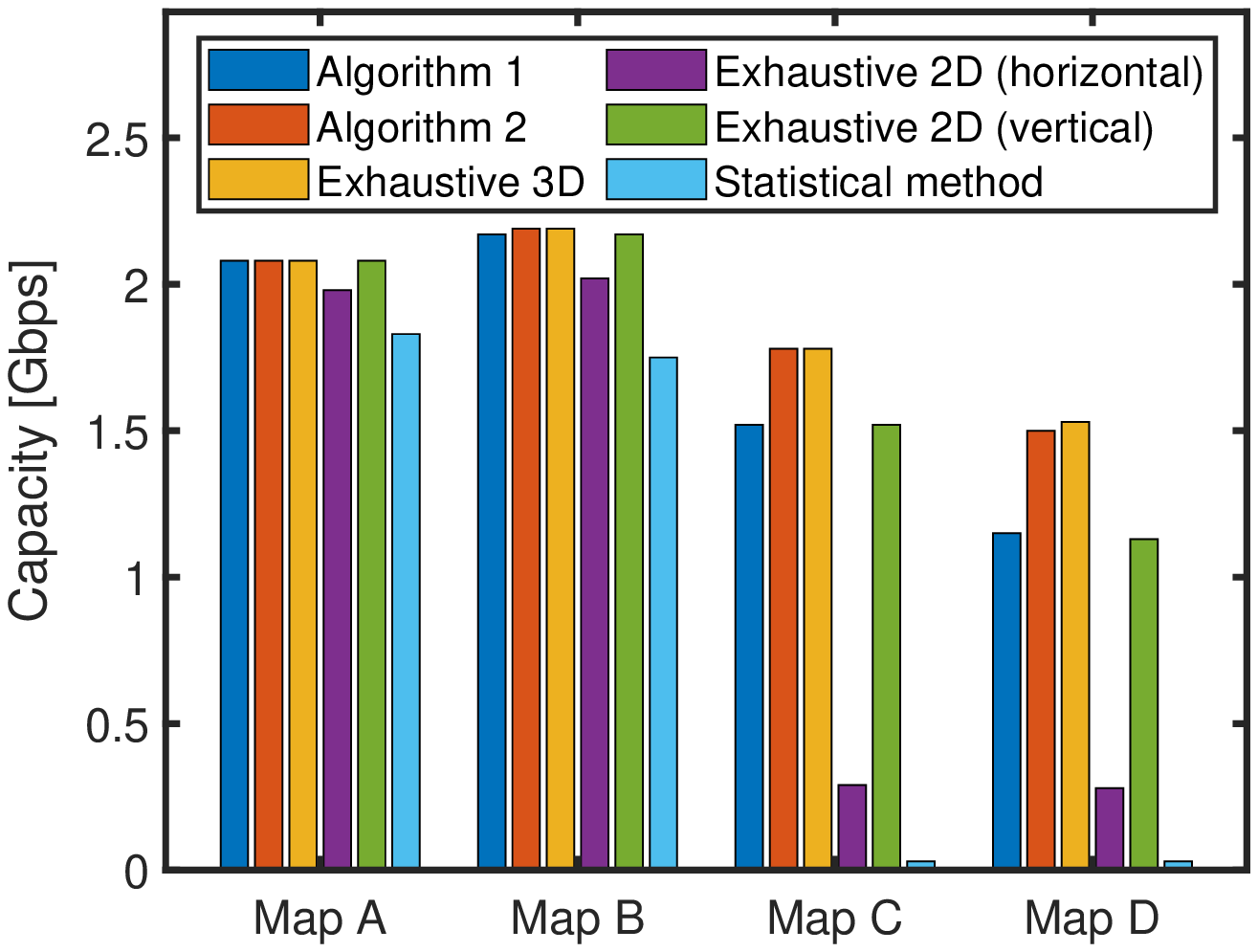}\caption{\label{fig:Capacity-on-diffenert-function-areas}Capacity in different
function areas}
\par\end{center}%
\end{minipage}%
\begin{minipage}[t]{0.5\columnwidth}%
\begin{center}
\includegraphics[width=1\columnwidth]{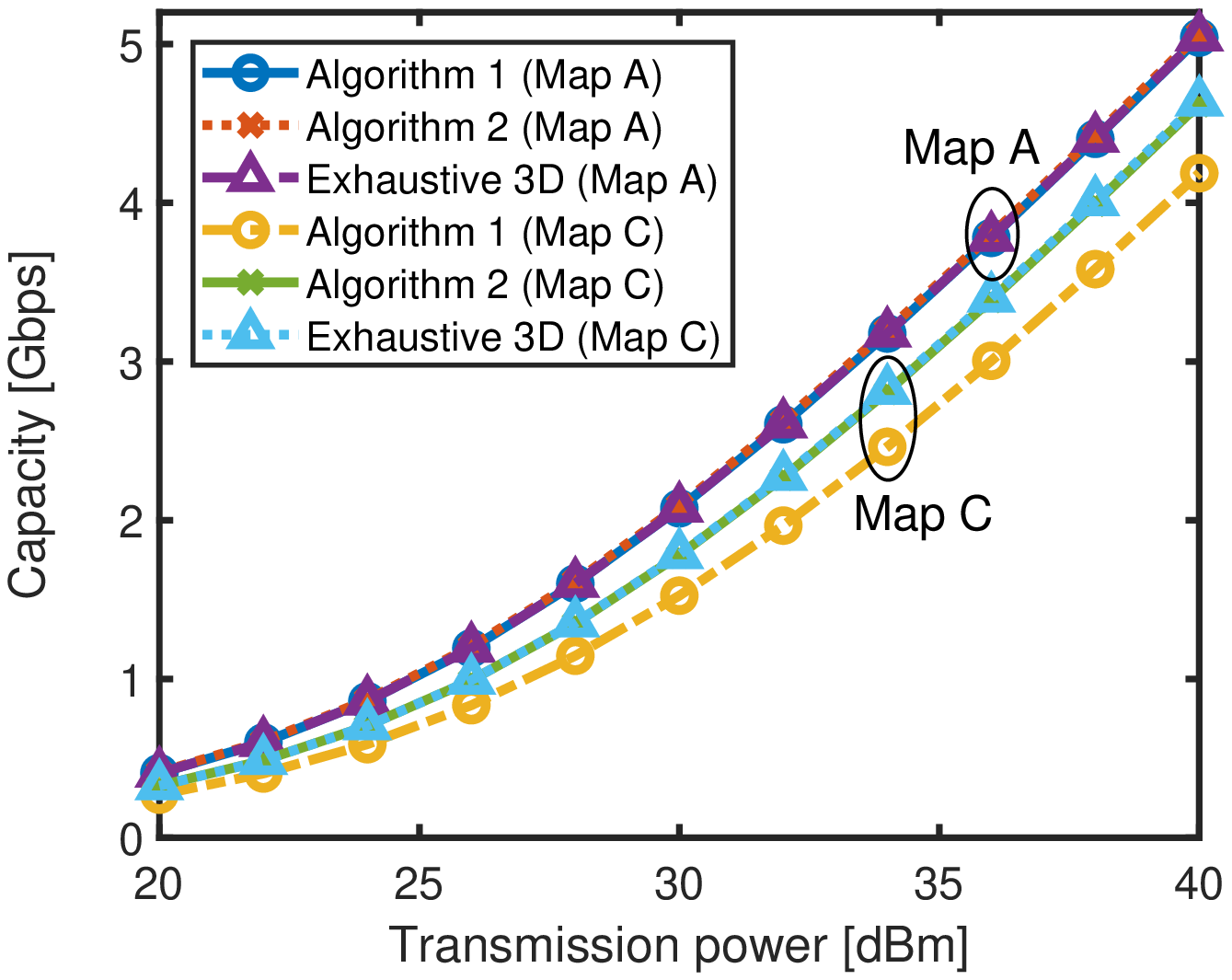}\caption{\label{fig:Capacity-under-different-power}Capacity under different
transmission power}
\par\end{center}%
\end{minipage}
\end{figure}

Fig. \ref{fig:Capacity-under-different-inter-user-distance} demonstrates
the average capacity versus different inter-user distance separating
the two ground users. It is observed that the capacity decreases as
the inter-user distance increases because increasing the inter-user
distance not only increases the propagation distance (resulting in
energy loss in free space), but also increases the chance of blockage,
and therefore, the UAV needs to fly higher to seek a double-LOS opportunity.
In addition, the performance gap between Algorithm~1 and the exhaustive
3D search scheme becomes smaller under larger inter-user distance.
In particular, Algorithm~1 achieves about $73\%$ to the exhaustive
3D scheme under the average inter-user distance of $105$~meters
and about $77\%$ under the average inter-user distance of $175$~meters
in Map D. Algorithm~2 with $\delta=3$ meters achieves above $97\%$
to the exhaustive 3D scheme under all inter-user distances both in
Map B and Map D.
\begin{figure}
\centering{}%
\begin{minipage}[t]{0.5\columnwidth}%
\begin{center}
\includegraphics[width=1\columnwidth]{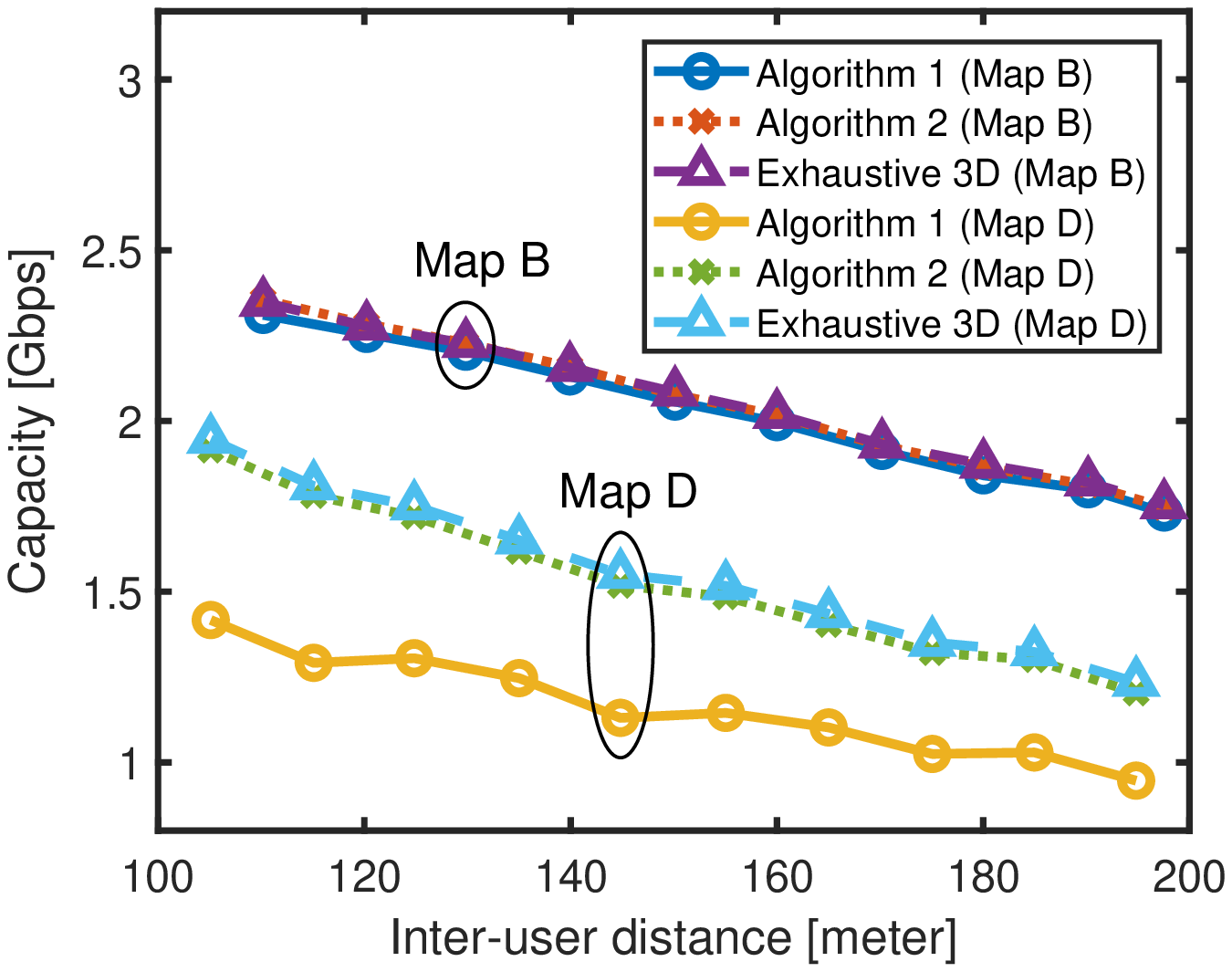}\caption{\label{fig:Capacity-under-different-inter-user-distance}Capacity
under different inter-user distance}
\par\end{center}%
\end{minipage}%
\begin{minipage}[t]{0.5\columnwidth}%
\begin{center}
\includegraphics[width=1\columnwidth]{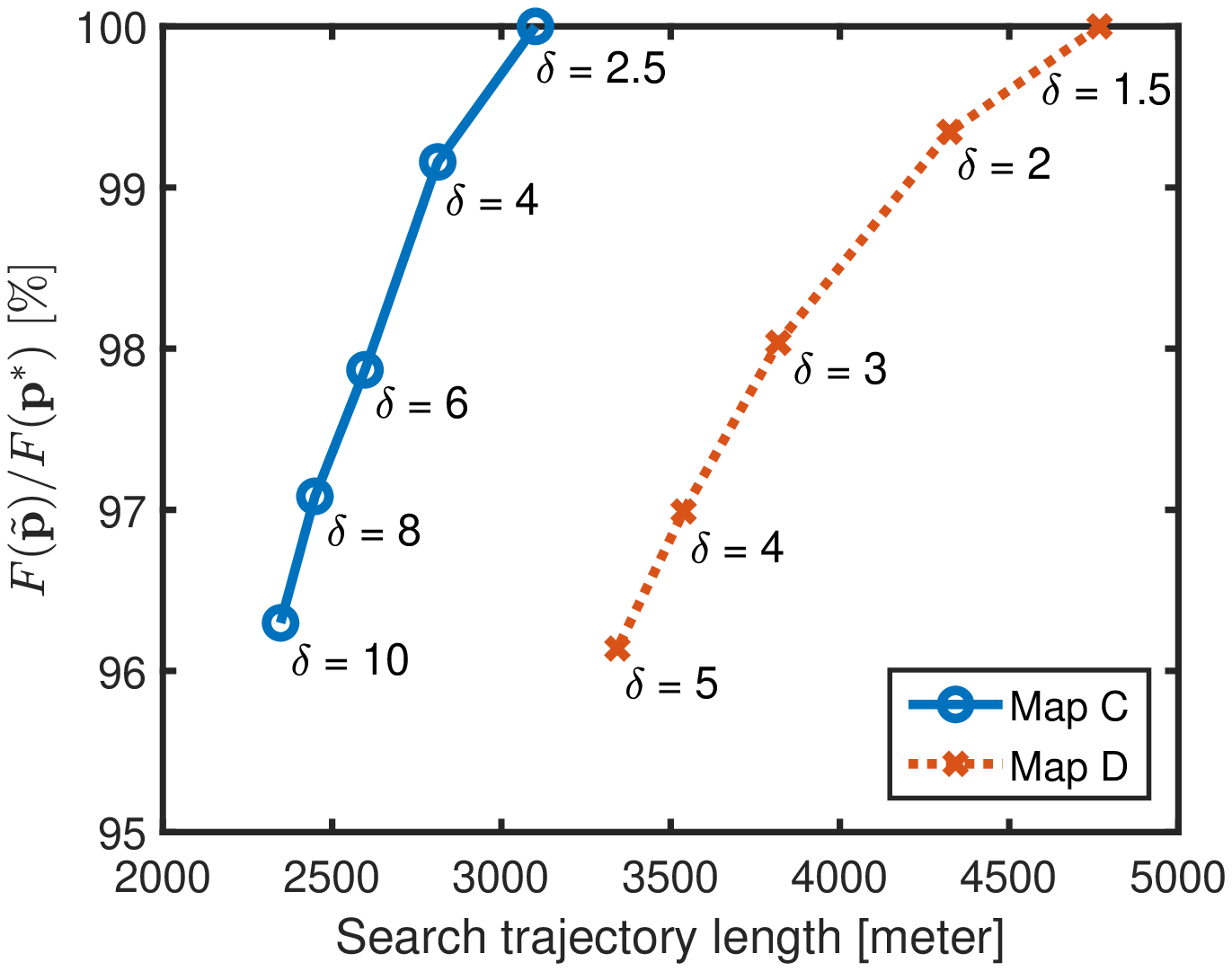}\caption{\label{fig:Comletion-time-optimality-trade-off}Trade-off of trajectory
length and optimality for Algorithm~2}
\par\end{center}%
\end{minipage}
\end{figure}

Fig. \ref{fig:Comletion-time-optimality-trade-off} evaluates the
performance-complexity trade-off of Algorithm~2 in terms of the percentage
of the performance $F(\mathbf{p}^{*})$ obtained from the Exhaustive
3D scheme, \emph{i.e.}, $F(\text{\ensuremath{\tilde{\mathbf{p}}}})/F(\mathbf{p}^{*})$,
over Maps C and D, under transmission power $P=30$ dBm and $M=4$.
With a choice of $\delta=4$ meters, Algorithm~2 achieves above $99.1\%$
to the Exhaustive 3D scheme in Map C within $2900$-meter search.
Additionally, it achieves above $99.3\%$ to the Exhaustive 3D scheme
within $4400$-meter search in Map D if $\delta$ is chosen as $2$
meters.

Finally, Table \ref{tab:comparison-of-mean-trajectory-length} summarizes
the mean trajectory length of different schemes on the four maps.
The search lengths of Algorithm~\ref{Alg:symmetric_algorithm} and
Algorithm~\ref{Alg:symmetric_algorithm_2} with $\delta=3$ meters
are merely several hundreds of meters for Maps A and B, but both of
the algorithms achieve above $99\%$ of the Exhaustive 3D search scheme
as seen from Fig.~\ref{fig:Capacity-on-diffenert-function-areas}.
The search lengths of Algorithm~\ref{Alg:symmetric_algorithm_2}
are about $3$-$4$ kilometers on Maps C and D to achieve above $98\%$
performance of the Exhaustive 3D scheme (see Fig.~\ref{fig:Comletion-time-optimality-trade-off}).
This corresponds to $3$-$4$ minutes flight time for a light-weight
commercial drone at a cruise speed of $20$ m/s.
\begin{table}
\caption{\label{tab:comparison-of-mean-trajectory-length}Comparison of mean
trajectory length {[}kilometer{]} under five schemes}

\renewcommand{\arraystretch}{1.2}
\centering{}%
\begin{tabular}{>{\centering}p{0.25\columnwidth}>{\centering}p{0.13\columnwidth}>{\centering}p{0.13\columnwidth}>{\centering}p{0.13\columnwidth}>{\centering}p{0.13\columnwidth}}
\hline 
\textbf{Scheme} & \textbf{Map A} & \textbf{Map B} & \textbf{Map C} & \textbf{Map D}\tabularnewline
\hline 
Algorithm~1 & 0.097 & 0.173 & 1.715 & 2.124\tabularnewline
Algorithm~2 & 0.272 & 0.569 & 3.002 & 3.818\tabularnewline
Exhaustive 2D (horizontal) & 9.353 & 11.3 & 185.5 & 191.5\tabularnewline
Exhaustive 2D (vertical) & 6.82 & 7.875 & 107 & 110\tabularnewline
Exhaustive 3D & 546 & 779.3 & 21,630 & 22,590\tabularnewline
\hline 
\end{tabular}
\end{table}

\subsection{Application in WPT}

\selectlanguage{english}%
Fig. \ref{fig:Harvested-power-in-WPT} shows the average harvested
power by the two ground users in the UAV WPT application \foreignlanguage{american}{under
transmission power $P=40$~dBm}. First, the relative performance
of Algorithm~1 to the global optimality (represented by that of the
Exhaustive 3D scheme) decreases as compared to the UAV relay communication
as shown in Fig. \ref{fig:Capacity-on-diffenert-function-areas}.
For example, Algorithm~1 achieves about $80.2\%$ to the \foreignlanguage{american}{Exhaustive
3D scheme} over Map C in the WPT application, but it achieves above
$85.3\%$ optimality in the relay communication application in Fig.
\ref{fig:Capacity-on-diffenert-function-areas}. This is because WPT
is more sensitive to the propagation distance as observed from its
objective function $f(d)$ defined in Section~\ref{subsec:Propagation-Environment-Modeling}.
However, Algorithm~2 with $\delta=3$ meters still achieves above
$98.9\%$ to the \foreignlanguage{american}{Exhaustive 3D scheme}
over all maps. Second, the performance gain of Algorithm~2 over Algorithm~1
is larger in the WPT application. For example, Algorithm~1 achieves
about $76.7\%$ of Algorithm~2 in Map D in the relaying application
while it achieves only $66.1\%$ of Algorithm~2 in WPT. Third, the
Exhaustive 2D scheme shows poor performance in dense areas (Maps C
and D) while it achieves over $90.3\%$ to the \foreignlanguage{american}{Exhaustive
3D scheme} in relatively sparse areas (Maps A and B). This is because
more blockage leads to fewer chances of finding a double-LOS position
close to both users on a fixed horizontal plane, and the power acquisition
efficiency decreases sharply with increasing distance when the \ac{uav}
flies away from the midpoint of the two users.\foreignlanguage{american}{}
\begin{figure}
\selectlanguage{american}%
\begin{centering}
\includegraphics[width=0.5\columnwidth]{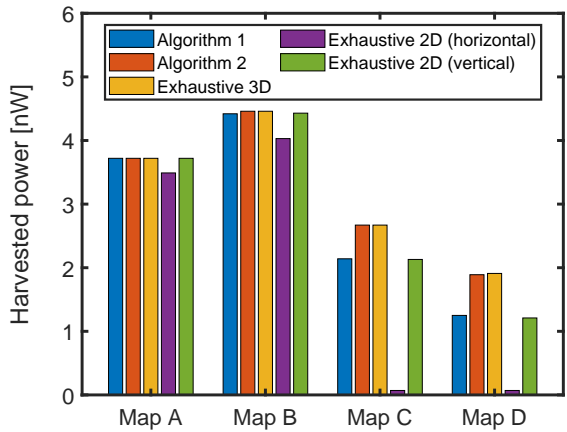}
\par\end{centering}
\centering{}\caption{\label{fig:Harvested-power-in-WPT}Harvested power in a WPT application}
\selectlanguage{english}%
\end{figure}

\selectlanguage{american}%

\section{Conclusions}

\label{sec:conclusion}

This paper developed two online search algorithms to search for the
globally optimal \ac{uav} position for establishing \ac{los} links
with two ground terminals in deep shadow. A key challenge addressed
here was to ensure \ac{los} conditions without the assistance of
3D maps. Exploiting the universal properties of any \ac{los} patterns
over an almost arbitrary terrain structure, Algorithm~1 found the
optimal position on the middle perpendicular plane with the search
length bounded by a linear function of the altitude of the initial
point. Algorithm~2 was proposed to search on a bounded 2D local area
for the $\epsilon$-optimal \ac{uav} position in 3D with search complexity
$O(1/\epsilon)$ under some mild condition. The optimality and complexity
were further confirmed by numerical experiments with real map data.
Both proposed algorithms achieved near 100\% global optimality over
several real city environments. Additionally, Algorithm~\ref{Alg:symmetric_algorithm_2}
achieves above $98.9\%$ \foreignlanguage{english}{performance of
the }Exhaustive 3D scheme over two simulated dense environments in
\foreignlanguage{english}{a WPT application.}


\appendices{}


\section{Proof of Theorem \ref{thm:Algorithm-1-global-optimal-on-S}}

\label{sec:Proof-of-Theorem1}

Suppose that the double-LOS trajectory $\hat{\mathbf{p}}(t)$ does
not terminate at the globally optimal solution $\hat{\mathbf{p}}$
of $\mathscr{P}'$. Contradiction can be shown in the following case-by-case
discussion.

Case 1: $\hat{\mathbf{p}}(t)$ passes by $\hat{\mathbf{p}}$ at time
$t_{1}$, but does not terminate at $\hat{\mathbf{p}}$. Since $\hat{\mathbf{p}}(t)$
only updates when meeting an double-\ac{los} position and each update
of $\hat{\mathbf{p}}(t)$ leads to a smaller radius to $\mathbf{o}$,
there exists another double-\ac{los} position $\hat{\mathbf{p}}(t_{2})$
such that $F(\hat{\mathbf{p}}(t_{2}))>F(\hat{\mathbf{p}}(t_{1}))=F(\hat{\mathbf{p}})$,
and $t_{2}>t_{1}$. This is a contradiction to the fact that $\hat{\mathbf{p}}$
is the globally optimal solution.

Case 2: $\hat{\mathbf{p}}(t)$ does not pass by $\hat{\mathbf{p}}$.
Suppose $\hat{\mathbf{p}}(t)$ terminates at $\hat{\mathbf{p}}(T_{1})$
and $\hat{\mathbf{p}}(T_{2})$ at the clockwise search stage and the
anticlockwise search stage, respectively. Since $\hat{\mathbf{p}}$
is the globally optimal solution of $\mathscr{P}'$, $F\text{(\ensuremath{\hat{\mathbf{p}}}(\ensuremath{T_{1}}))}\leq F(\hat{\mathbf{p}})$
and $F\text{(\ensuremath{\hat{\mathbf{p}}}(\ensuremath{T_{2}}))}\leq F(\hat{\mathbf{p}})$
result in $r(\hat{\mathbf{p}}(T_{1}))\geq r(\hat{\mathbf{p}})$ and
$r(\hat{\mathbf{p}}(T_{2}))\geq r(\hat{\mathbf{p}})$ according to
Lemma~\ref{lem:minimize-the-radius}. Thus, $\hat{\mathbf{p}}$ is
embraced by the search trajectory, which implies that the search trajectory
$\mathbf{p}(t)$ must pass a point $\mathbf{p}'$ that is perpendicularly
above $\hat{\mathbf{p}}$. As Lemma~\ref{lem:Double-LOS-structure-on-S}
shows that any position perpendicularly above $\hat{\mathbf{p}}$
is also double-LOS, then, the line segment joining $\hat{\mathbf{p}}$
and $\mathbf{p}'$ must be double-LOS. However, according to Step~\ref{enu:Clockwise_if}
of Algorithm~1, the \ac{uav} must go downwards until reaching $\hat{\mathbf{p}}$,
leading to a contradiction for Case 2.

To summarize, the double-LOS trajectory $\hat{\mathbf{p}}(t)$ must
terminate at the globally optimal solution to $\mathscr{P}'$.

\section{Proof of Proposition \ref{thm:infer_given_two_positions}}

\label{sec:Proof-of-Theorem-double-LOS-line}

Denote the plane perpendicular to the ground and passing through points
$\mathbf{u}_{i}$ and $\mathbf{a}_{i}$ as $\pi_{i}$, $i\in\left\{ 1,2\right\} $.
Define the intersection line of $\pi_{1}$ and $\pi_{2}$ as $l_{12}$.
Given $a_{11}a_{21}>0$, define a position $\tilde{\mathbf{a}}_{i}$
on $l_{12}$ such that $\mathbf{u}_{i}$, $\mathbf{a}_{i}$, and $\tilde{\mathbf{a}}_{i}$
are colinear and $\tilde{a}_{i3}>0$. The colinear invariant property
of the \ac{los} regions implies that $\tilde{\mathbf{a}}_{i}$ is
\ac{los} from user $i$, {\em i.e.}, $\tilde{\mathbf{a}}_{i}\in\mathcal{D}_{0}^{(i)}$,
given $\mathbf{a}_{i}\in\mathcal{D}_{0}^{(i)}$. Then, by applying
the upward invariant property of the \ac{los} regions, one can obtain
$\tilde{\mathbf{a}}_{1}$ is a double-LOS position if $\tilde{\mathbf{a}}_{1}$
is higher than $\tilde{\mathbf{a}}_{2}$, \emph{i.e.}, $a_{13}/a_{23}>a_{11}/a_{21}$.
Otherwise, $\tilde{\mathbf{a}}_{2}$ is a double-LOS position.

The global optimality can be proved by contradictions. Suppose there
exists another double-LOS point $\tilde{\mathbf{a}}$ off $l_{12}$
or there exists another double-LOS point $\tilde{\mathbf{a}}$ lower
than both $\tilde{\mathbf{a}}_{1}$ and $\tilde{\mathbf{a}}_{2}$
on $l_{12}$. If so, the intersection point between the middle perpendicular
plane and the line joining $\mathbf{u}_{i}$ and $\tilde{\mathbf{a}}$
will be off the double-ray \ac{los} pattern. This is a contradiction
to the unique existence of double-ray \ac{los} pattern. Thus, either
$\tilde{\mathbf{a}}_{1}$ or $\tilde{\mathbf{a}}_{2}$ is the globally
optimal solution to $\mathscr{P}$ since one of them is the lowest
double-LOS point on $l_{12}$.

By applying the knowledge of analytic geometry, one can obtain $\tilde{\mathbf{a}}_{1}$
and $\tilde{\mathbf{a}}_{2}$ as
\[
\tilde{\mathbf{a}}_{1}=\frac{2a_{21}}{a_{11}+a_{21}}\mathbf{a}_{1}-\frac{L}{2}\mathbf{e}_{2},\text{ and }\tilde{\mathbf{a}}_{2}=\frac{2a_{11}}{a_{11}+a_{21}}\mathbf{a}_{2}+\frac{L(a_{21}-2a_{11})}{2a_{11}}\mathbf{e}_{2},
\]
respectively, where $\mathbf{e}_{2}=(\mathbf{u}_{2}-\mathbf{u}_{1})/\|\mathbf{u}_{2}-\mathbf{u}_{1}\|_{2}$,
and $L=\|\mathbf{u}_{2}-\mathbf{u}_{1}\|_{2}$.

\section{The Closed-form Solution to $\mathscr{P}''$}

\label{sec:The-Closed-form-Solution}

Without loss of generality, one only needs to consider the case $b_{21}\geq a_{21}>b_{11}\geq a_{11}>0$
due to the symmetric properties, and this case offers $\max\left\{ d_{1}(\mathbf{p}),d_{2}(\mathbf{p})\right\} =d_{1}(\mathbf{p})$.
Since $F(\mathbf{p})$ is decreasing with $\max\left\{ d_{1}(\mathbf{p}),d_{2}(\mathbf{p})\right\} $,
maximizing $F(\mathbf{p})$ will be equivalent to minimizing $d_{1}(\mathbf{p})$.
According to Proposition~\ref{thm:infer_given_two_positions}, if
$x_{13}/x_{23}\geq x_{11}/x_{21}$, $\mathbf{q}_{1}(\mathbf{x}_{1},\mathbf{x}_{2})=2x_{21}/(x_{11}+x_{21})\mathbf{x}_{1}-L/2\mathbf{e}_{2}$
will be the 3D globally optimal solution to $\mathscr{P}$ given the
\ac{los} pair $(\mathbf{x}_{1},\mathbf{x}_{2})$ for a double-ray
pattern. Otherwise, $\mathbf{q}_{2}(\mathbf{x}_{1},\mathbf{x}_{2})=2x_{11}/(x_{11}+x_{21})\mathbf{x}_{2}+L(x_{21}-2x_{11})/(2x_{11})\mathbf{e}_{2}$
will be the solution. Here, we consider the case $x_{13}/x_{23}\geq x_{11}/x_{21}$,
and the other case is similar to it. Given the above conditions, $\mathscr{P}''$
can be transformed as the following problem.
\begin{equation}
\begin{aligned}\mathop{\mbox{minimize}}\limits _{\mathbf{q}_{1}(\mathbf{x}_{1},\mathbf{x}_{2})} & \quad d_{1}(\mathbf{q}_{1}(\mathbf{x}_{1},\mathbf{x}_{2}))\\
\mathop{\mbox{subject to}} & \quad a_{i1}\leq x_{i1}\leq b_{i1}\\
 & \quad x_{11}\leq x_{21}x_{13}/x_{23}\\
 & \quad\mathbf{x}_{1}=(x_{11},a_{12},h_{1}),\mathbf{x}_{2}=(x_{21},a_{22},h_{2}).
\end{aligned}
\label{eq:problem_transformed_case1}
\end{equation}

The objective function in (\ref{eq:problem_transformed_case1}) is
decreasing with $x_{21}$, and it has only one stationary point with
$x_{11}$. In addition, the feasible domains of $x_{11}$ and $x_{21}$
in (\ref{eq:problem_transformed_case1}) are bounded while other variables
are constants. Then the solving process can be summarized as the following
four steps.
\begin{enumerate}
\item Pick the optimal value $x_{21}^{*}(x_{11})$ from $x_{21}\in[\max\left\{ a_{21},x_{11}h_{2}/h_{1}\right\} ,b_{21}]$
through the monotonicity of $d_{1}(\mathbf{q}_{1}(\mathbf{x}_{1},\mathbf{x}_{2}))$.
\item Given $x_{21}^{*}(x_{11})$, calculate the stationary point of $d_{1}(\mathbf{q}_{1}(\mathbf{x}_{1},\mathbf{x}_{2}))$
over $x_{11}$, and obtain the optimal value $x_{11}^{*}$ from $x_{11}\in[a_{11},\min\left\{ b_{11},x_{21}^{*}(x_{11})\cdot h_{1}/h_{2}\right\} ]$.
\item Calculate the optimal value of $x_{21}$ as $x_{21}^{*}=x_{21}^{*}(x_{11}^{*})$.
\item The solution $\mathbf{q}_{1}(\mathbf{x}_{1}^{*},\mathbf{x}_{2}^{*})$
can be obtained by substituting $\mathbf{x}_{1}^{*}=x_{11}^{*}\mathbf{e}_{1}+a_{12}\mathbf{e}_{2}+h_{1}\mathbf{e}_{3}$
and $\mathbf{x}_{2}^{*}=x_{21}^{*}\mathbf{e}_{1}+a_{22}\mathbf{e}_{2}+h_{2}\mathbf{e}_{3}$
into formula (\ref{eq:solution_two_positions}).
\end{enumerate}

The similar method can be applied to other cases. Finally, the solution
to $\mathscr{P}''$ is given as
\begin{equation}
Q(\mathbf{a}_{1},\mathbf{b}_{1};\mathbf{a}_{2},\mathbf{b}_{2})=\begin{cases}
q_{1}(\rho_{1})\quad\text{if \ensuremath{|b_{21}|\geq|a_{21}|>|b_{11}|\geq|a_{11}|}\text{ and }}\ensuremath{\frac{h_{2}}{h_{1}}\leq\frac{b_{21}}{a_{11}}}\\
q_{2}(\rho_{2})\quad\text{if \ensuremath{|b_{21}|\geq|a_{21}|>|b_{11}|\geq|a_{11}|}}\text{\text{ and }}\frac{h_{2}}{h_{1}}>\frac{b_{21}}{a_{11}}\\
q_{2}(\rho_{3})\quad\text{if \ensuremath{|b_{11}|\geq|a_{11}|>|b_{21}|\geq|a_{21}|}}\text{\text{ and }}\ensuremath{\frac{h_{2}}{h_{1}}\geq\frac{a_{11}}{b_{21}}}\\
q_{1}(\rho_{4})\quad\text{if \ensuremath{|b_{11}|\geq|a_{11}|>|b_{21}|\geq|a_{21}|}}\text{\text{ and }}\ensuremath{\frac{h_{2}}{h_{1}}<\frac{a_{11}}{b_{21}}}
\end{cases}\label{eq:solution_problem-double-stripe}
\end{equation}
where
\begin{align*}
\mathbf{q}_{1}(\rho) & =\frac{2a_{11}\rho}{\rho+h_{1}}\mathbf{e}_{1}+\frac{(\rho-h_{1})L}{2(\rho+h_{1})}\mathbf{e}_{2}+\frac{2h_{1}\rho}{\rho+h_{1}}\mathbf{e}_{3},\,\mathbf{q}_{2}(\rho)=\frac{2a_{21}\rho}{\rho+h_{2}}\mathbf{e}_{1}+\frac{(h_{2}-\rho)L}{2(\rho+h_{2})}\mathbf{e}_{2}+\frac{2h_{2}\rho}{\rho+h_{2}}\mathbf{e}_{3},\\
\rho_{1} & =\max\left\{ h_{2},\frac{a_{21}h_{1}}{a_{11}}\right\} ,\,\rho_{2}=\text{median}\left\{ \frac{a_{11}h_{2}}{a_{21}},\frac{b_{11}h_{2}}{a_{21}},h_{1},h_{2},\frac{L^{2}h_{2}}{4a_{21}^{2}+4h_{2}^{2}}\right\} ,\\
\rho_{3} & =\max\left\{ h_{1},\frac{a_{11}h_{2}}{a_{21}}\right\} ,\,\rho_{4}=\text{median}\left\{ \frac{a_{21}h_{1}}{a_{11}},\frac{b_{21}h_{1}}{a_{11}},h_{1},h_{2},\frac{L^{2}h_{1}}{4a_{11}^{2}+4h_{1}^{2}}\right\} .
\end{align*}

\section{Proof of Lemma \ref{thm:search_area_to_find_optimalp}}

\label{sec:Proof-of-Theorem-search-area}

First, it can be easily verified that $\mathcal{B}(\mathbf{p}_{0})\subseteq\tilde{\mathcal{B}}(\mathbf{p}_{0})$.
It follows that, given a double-LOS point $\mathbf{p}_{0}$, the optimal
solution $\mathbf{p}^{*}$ to $\mathscr{P}$ must lie in $\tilde{\mathcal{B}}(\mathbf{p}_{0})$
due to the fact that the objective function is decreasing in $d_{0}(\mathbf{p})$.

Then, for each user $i\in\{1,2\}$, find the points $\mathbf{p}_{i}\in\mathcal{S}$
and $\mathbf{p}_{i}'\in\mathcal{H}$, such that the points $\mathbf{u}_{i}$,
$\mathbf{p}^{*}$, $\mathbf{p}_{i}$, and $\mathbf{p}_{i}'$ are colinear
as illustrated in Fig.~\ref{fig:region_B}. It is clear that either
$\mathbf{p}_{i}$ or $\mathbf{p}_{i}'$ belongs to $\tilde{\mathcal{B}}(\mathbf{p}_{0})\cap\left\{ \mathcal{S}\cup\mathcal{H}\right\} $.
Specifically, for $d_{0}(\mathbf{p}_{0})\leq\sqrt{2}L/2$, we have
$\mathbf{p}_{i}\in\mathcal{B}(\mathbf{p}_{0})\cap\mathcal{S}$ or
$\mathbf{p}_{i}'\in\mathcal{B}(\mathbf{p}_{0})\cap\mathcal{H}$; for
$d_{0}(\mathbf{p}_{0})\geq L$, we have $\mathbf{p}_{i}'\in\mathcal{B}(\mathbf{p}_{0})\cap\mathcal{H}$;
and for $\sqrt{2}L/2<d_{0}(\mathbf{p}_{0})<L$, we have $\mathbf{p}_{i}\in\tilde{\mathcal{B}}(\mathbf{p}_{0})\cap\mathcal{S}$
or $\mathbf{p}_{i}'\in\tilde{\mathcal{B}}(\mathbf{p}_{0})\cap\mathcal{H}$.

Due to the colinear invariant property for the \ac{los} regions,
it suffices to search either $\mathbf{p}_{i}$ or $\mathbf{p}_{i}'$
for the \ac{los} status to user $i$, where both $\mathbf{p}_{i}$
and $\mathbf{p}_{i}'$ have the same \ac{los} status.

As a result, given the \ac{los} status found for $\mathbf{p}_{1}$
and $\mathbf{p}_{2}$, Proposition~\ref{thm:infer_given_two_positions}
asserts that the optimal solution can be found as $\mathbf{p}^{*}=\mathbf{q}(\mathbf{p}_{1},\mathbf{p}_{2})$
in (\ref{eq:solution_two_positions}). Since either $\mathbf{p}_{i}$
or $\mathbf{p}_{i}'$ belongs to $\tilde{\mathcal{B}}(\mathbf{p}_{0})\cap\left\{ \mathcal{S}\cup\mathcal{H}\right\} $
and the two points are related according to (\ref{eq:transition_from_H_S1})
and (\ref{eq:transition_from_H_S2}), it thus suffices to search $\tilde{\mathcal{B}}(\mathbf{p}_{0})\cap\left\{ \mathcal{S}\cup\mathcal{H}\right\} $
for computing $\mathbf{p}^{*}$. The result of Lemma~\ref{thm:search_area_to_find_optimalp}
is thus proven.

\section{Proof of Theorem \ref{thm:Alg2-Upper-bound-of-the-performance-gap-1}}

\label{sec:Proof-of-Theorem-upper-bound-performance-gap-Alg2}
\begin{figure}
\begin{centering}
\includegraphics[width=0.5\linewidth]{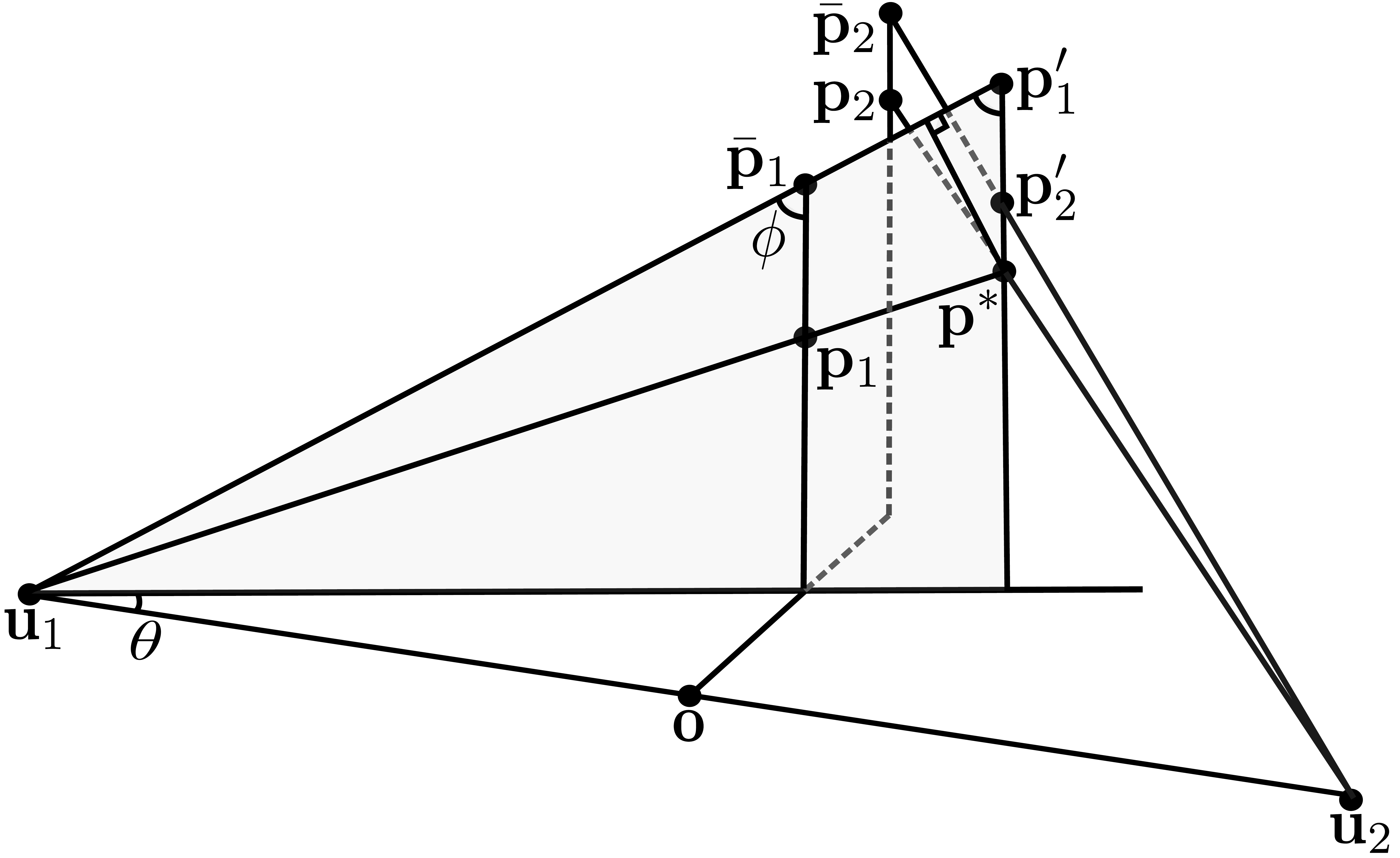}
\par\end{centering}
\caption{Performance gap due to the vertical step size $\delta$.}

\label{fig:performancegap_case1}
\end{figure}
Consider the region $\mathcal{B}'(\tilde{\mathbf{p}})\triangleq\{\mathbf{p}\in\mathbb{R}^{3}:p_{3}\geq H'_{\text{min}}(\tilde{\mathbf{p}}),d_{0}(\mathbf{p})\leq d_{0}(\tilde{\mathbf{p}})\}$.
Since $\tilde{\mathbf{p}}$ as the solution to problem (\ref{eq:problem_found_from_trajectory})
is a double-LOS point and $H'_{\text{min}}(\tilde{\mathbf{p}})<H_{\text{min}}$,
we must have $\mathbf{p}^{*}\in\mathcal{B}'(\tilde{\mathbf{p}})$.
In addition, $\mathcal{B}'(\tilde{\mathbf{p}})\subseteq\mathcal{B}'(\mathbf{p}_{0})$.
Without loss of generality, consider the case $\text{\ensuremath{d_{0}(\mathbf{p}^{*})=}}d_{1}(\mathbf{p}^{*})>d_{2}(\mathbf{p}^{*})$
which means that $\mathbf{p}^{*}$ is closer to $\mathbf{u}_{2}$
as shown in Fig.~\ref{fig:performancegap_case1}.

First, given $d_{0}(\tilde{\mathbf{p}})\leq\sqrt{2}L/2$, it holds
from the geometry that the colinear point $\mathbf{p}_{1}$, \ac{wrt}
$\mathbf{u}_{1}$ and $\mathbf{p}^{*}$, lies in the effective search
region $\mathcal{B}'(\mathbf{p}_{0})\cap\mathcal{S}$. Then, define
$\bar{\mathbf{p}}_{1}\in\mathcal{T}$ as a point on trajectories such
that $\bar{\mathbf{p}}_{1}$ is perpendicularly above $\mathbf{p}_{1}$
and the closest to $\mathbf{p}_{1}$. It is clear that $\|\bar{\mathbf{p}}_{1}-\mathbf{p}_{1}\|_{2}\leq\delta$,
since the trajectories are parallel to the ground with $\delta$ space.
Additionally, $\bar{\mathbf{p}}_{1}$ is \ac{los} to $\mathbf{u}_{1}$
due to the upward invariant property. Similarly, one can define $\mathbf{p}_{2}$
and $\bar{\mathbf{p}}_{2}$.

Next, define $\mathbf{p}_{1}'$ as a point perpendicularly above $\mathbf{p}^{*}$
such that $\mathbf{u}_{1}$, $\bar{\mathbf{p}}_{1}$, and $\mathbf{p}_{1}'$
are colinear. Similarly, one can obtain the colinear point $\mathbf{p}_{2}'$.
Given the \ac{los} status of $\bar{\mathbf{p}}_{1}$ and $\bar{\mathbf{p}}_{2}$,
Proposition~\ref{thm:infer_given_two_positions} asserts that either
$\mathbf{p}_{1}'=\mathbf{q}(\bar{\mathbf{p}}_{1},\bar{\mathbf{p}}_{2})$
or $\mathbf{p}_{2}'=\mathbf{q}(\bar{\mathbf{p}}_{1},\bar{\mathbf{p}}_{2})$
is a double-LOS point, and the other one is a non-double-LOS point.
Since $\mathbf{q}(\bar{\mathbf{p}}_{1},\bar{\mathbf{p}}_{2})$ is
a suboptimal solution to problem (\ref{eq:problem_found_from_trajectory}),
the double-LOS point $\mathbf{p}_{1}'$ or $\mathbf{p}_{2}'$ must
lie in the feasible set of problem (\ref{eq:problem_found_from_trajectory}).
As $\tilde{\mathbf{p}}$ is the solution to problem (\ref{eq:problem_found_from_trajectory}),
it holds that $d_{0}(\mathbf{p}_{1}')\geq d_{0}(\tilde{\mathbf{p}})$
if $\mathbf{p}_{1}'$ is double-LOS, and $d_{0}(\mathbf{p}_{2}')\geq d_{0}(\tilde{\mathbf{p}})$,
otherwise.

Without loss of generality, consider $\mathbf{p}_{1}'$ is double-LOS.
Then, there exists a point $\tilde{\mathbf{p}}'$ perpendicularly
above $\mathbf{p}^{*}$ and below $\mathbf{p}_{1}'$ such that $d_{0}(\mathbf{p}_{1}')\geq d_{0}(\tilde{\mathbf{p}}')=d_{0}(\tilde{\mathbf{p}})\geq d_{0}(\mathbf{p}^{*})$.
Define $\phi$ as the angle between $\mathbf{u}_{1}-\mathbf{p}_{1}'$
and $\mathbf{p}^{*}-\mathbf{p}_{1}'$, and define $\theta$ as the
angle between the basis vector $\mathbf{e}_{2}$ and the plane containing
$\mathbf{u}_{1}$, $\mathbf{p}_{1}$, and $\bar{\mathbf{p}}_{1}$.
Based on the geometric properties, one can obtain
\[
d_{0}(\tilde{\mathbf{p}})-d_{0}(\mathbf{p}^{*})\leq\|\tilde{\mathbf{p}}'-\mathbf{p}^{*}\|_{2}\cos\phi\leq\frac{2}{L}\|\bar{\mathbf{p}}_{1}-\mathbf{p}_{1}\|_{2}d_{0}(\mathbf{\tilde{\mathbf{p}}})\cos\phi\cos\theta\sin\phi.
\]

Given $\cos\theta\leq1$, $\cos\phi\leq1$ and $p_{13}\geq H'_{\text{min}}(\tilde{\mathbf{p}})$,
the upper bound of $\sin\phi$ is given by
\[
\sin\phi\leq\frac{L}{2\cos\theta\sqrt{(L/(2\cos\theta))^{2}+(p_{13}+\|\bar{\mathbf{p}}_{1}-\mathbf{p}_{1}\|_{2})^{2}}}\leq\frac{L}{2\cos\theta\sqrt{(L/2)^{2}+(H'_{\text{min}}(\tilde{\mathbf{p}}))^{2}}}.
\]

Given $\|\bar{\mathbf{p}}_{1}-\mathbf{p}_{1}\|_{2}\leq\delta$ and
$\cos\phi\leq1$, the gap $d_{0}(\tilde{\mathbf{p}})-d_{0}(\mathbf{p}^{*})$
is upper bounded by
\begin{equation}
d_{0}(\tilde{\mathbf{p}})-d_{0}(\mathbf{p}^{*})\leq\frac{2}{L}\|\bar{\mathbf{p}}_{1}-\mathbf{p}_{1}\|_{2}d_{0}(\tilde{\mathbf{p}})\cos\phi\cos\theta\sin\phi\leq\frac{2}{L}\delta\sqrt{d_{0}^{2}(\tilde{\mathbf{p}})-H_{\text{min}}^{2}}.\label{eq:distance_gap}
\end{equation}

Similarly, one can derive the same gap as (\ref{eq:distance_gap})
if $\mathbf{p}_{2}'$ is a double-LOS point. Therefore, the distance
gap between $d_{0}(\tilde{\mathbf{p}})$ and $d_{0}(\mathbf{p}^{*})$
is upper bounded by $2\delta\sqrt{d_{0}^{2}(\tilde{\mathbf{p}})-H_{\text{min}}^{2}}/L$.
As $f(d)$ is decreasing with $d$, $f'(d)<0$, and thus, $f'(d_{0}(\mathbf{p}^{*}))<0$.
Then, the first order condition of the convex function $f(d)$ shows
that $f(d_{0}(\tilde{\mathbf{p}}))-f(d_{0}(\mathbf{p}^{*}))\text{\ensuremath{\geq}}2\delta f'(d_{0}(\mathbf{p}^{*}))\sqrt{d_{0}^{2}(\tilde{\mathbf{p}})-H_{\text{min}}^{2}}/L.$

One can obtain $F(\mathbf{p})=\min\left\{ f(d_{1}(\mathbf{p})),f(d_{2}(\mathbf{p}))\right\} =f(d_{0}(\mathbf{p}))$
using the definition of $d_{0}(\mathbf{p})$ and the monotonicity
of $f(d_{i}(\mathbf{p}))$. Hence, the performance gap can be obtained
as
\[
F(\mathbf{p}^{*})-F(\tilde{\mathbf{p}})=f(d_{0}(\mathbf{p}^{*}))-f(d_{0}(\tilde{\mathbf{p}}))\leq-\frac{2}{L}\delta f'(d_{0}(\mathbf{p}^{*}))\sqrt{d_{0}^{2}(\tilde{\mathbf{p}})-H_{\text{min}}^{2}}.
\]

\bibliographystyle{IEEEtran}
\bibliography{IEEEabrv,StringDefinitions,JCgroup,ChenBibCV}

\end{document}